\documentclass[12pt]{article}
\pdfoutput=1
\usepackage{comment}
\usepackage{soul}
\usepackage{geometry}
\usepackage{changepage}
\usepackage{amsmath, amssymb, bbm, amsthm, mathabx, bigints}

\newtheorem{theorem}{Theorem}

\newcommand{\He}[2]{\text{\normalfont He}_{#1}\!\left(#2\right)}
\usepackage{float}
\usepackage{setspace}
\usepackage[font=small]{caption}
\usepackage[colorlinks,citecolor=DarkBlue,linkcolor=DarkBlue,bookmarks=false,hypertexnames=true, urlcolor=blue]{hyperref} 
\usepackage[svgnames]{xcolor} 
\usepackage{stfloats} 
\usepackage[american]{babel}

\usepackage[affil-it]{authblk} 
\usepackage{bbm}

\usepackage{fancyhdr}
\usepackage[font=footnotesize,labelfont=sc]{caption}
\usepackage{accents}

\renewcommand{\thesection}{\Roman{section}} 
\usepackage{titlesec}
\usepackage{lipsum}
\usepackage{lmodern}

\renewcommand{\thesection}{\arabic{section}}
\renewcommand{\thesubsection}{\thesection.\arabic{subsection}}
\renewcommand{\thesubsubsection}{\thesubsection.\arabic{subsubsection}}

\titleformat{\section}
  {\centering\scshape\normalsize}
  {\thesection. }{0.1em}{}

\titleformat{\subsection}{\centering\normalsize\itshape}{\thesubsection\ }{0.1em}{}

\titleformat{\subsubsection}[runin]{\normalsize\itshape}{\thesubsubsection}{1em}{}

\usepackage{graphicx}
\usepackage{grffile}
\usepackage{tablefootnote}
\usepackage{dcolumn}
\usepackage{tabularx, afterpage}

\usepackage{pdflscape}
\usepackage{pdfpages}

\usepackage{footnote}
\makesavenoteenv{tabular}
\usepackage[multiple]{footmisc}

\usepackage[toc,page, title]{appendix}

\usepackage{csquotes}

\newcommand{\independent}{\perp \!\!\! \perp} 

\usepackage{natbib}

\renewenvironment{abstract}
 {\small
  \begin{center}
  \bfseries \abstractname\vspace{-.5em}\vspace{0pt}
  \end{center}
  \list{}{%
    \setlength{\leftmargin}{21mm}
    \setlength{\rightmargin}{\leftmargin}%
  }%
  \item\relax}
 {\endlist}

\theoremstyle{plain}

\newtheorem{assumption}{Assumption}

\newtheorem{remark}{Remark}

\newtheorem{lemma}{Lemma}

\newtheorem{corollary}{Corollary}

\newtheorem{recommendation}{Recommendation}

\usepackage{chngcntr}
\usepackage{apptools}
\AtAppendix{\counterwithin{lemma}{section}} 
\AtAppendix{\counterwithin{assumption}{section}}
\AtAppendix{\counterwithin{theorem}{section}}

\begin{document}

\makeatletter
\patchcmd{\@maketitle}{\LARGE \@title}{\fontsize{12}{12}\selectfont\@title}{}{}
\newcommand\customsize{\@setfontsize\customsize{8.5}{11}}
\makeatother

\renewcommand\Authfont{\fontsize{12}{12}\selectfont} 

\newcommand\blfootnote[1]{%
  \begingroup
  \renewcommand\thefootnote{}\footnote{#1}%
  \addtocounter{footnote}{-1}%
  \endgroup
}

\font\myfont=cmr12 at 16pt
\title{{\LARGE How Replicable Are Statistically Significant Findings?\thanks{
We thank Jonathan Roth and seminar participants at the University of Queensland and the University of Western Australia for helpful comments and suggestions.
}}}

\author[1]{\large Patrick Vu\thanks{Department of Economics, University of New South Wales. Email: \href{[patrick.vu@unsw.edu.au]}{patrick.vu@unsw.edu.au}} and
Stefan Faridani\thanks{Georgia Institute of Technology. Email: \href{[sfaridani6@gatech.edu]}{sfaridani6@gatech.edu}}.}

\date{\today}

\renewcommand{\thefootnote}{\fnsymbol{footnote}}
\setcounter{footnote}{0}   
\maketitle
\thispagestyle{empty}
\renewcommand{\thefootnote}{\arabic{footnote}}
\setcounter{footnote}{0}

\maketitle
\begin{abstract}
In the empirical sciences, significance thresholds often determine whether findings are treated as evidence of an effect. This paper studies how likely findings that just meet conventional significance thresholds are to remain significant in replications of the same sample size. To answer this question, we estimate the expected replication probability conditional on a given $p$-value among published studies for experimental economics, psychology, and social science. We validate this measure by showing it accurately predicts actual replication outcomes, outperforming prediction markets. A finding with a $p$-value of 0.05 has an expected replication probability ranging from 0.10 to 0.25 across fields. Low replicability reflects low power in original studies rather than publication bias. We then develop a nonparametric estimator and apply it to economics literatures that use larger samples, finding higher but still low replication probabilities. These results indicate that statistical significance in a single study provides only suggestive evidence of an effect. Stronger conclusions require cumulative evidence.
\end{abstract}


\newpage
\setcounter{page}{1}

\section{Introduction}
In the empirical sciences, conventional significance thresholds are widely used to establish evidence of an effect. In some cases, a statistically significant empirical claim is repeated using the same methodology and a new sample \citep{OpenScience2015, Camerer2016, Camerer2018}. A significant replication is deemed `successful' and strengthens the original claim, while an insignificant replication is deemed `unsuccessful' and weakens it.



Given the ubiquity of significance testing in applied research, this paper asks a simple but fundamental question: ``How likely are findings that just meet conventional significance thresholds to remain significant in replications of the same sample size?'' If significant findings are unlikely to remain significant in replication, then the initial result may provide weaker grounds for claiming a nonzero effect than is commonly supposed. This concern differs from critiques that significance testing obscures economic significance or encourages $p$-hacking. Instead, in this paper, we evaluate significance testing on its own terms.


To answer this question, we derive the predictive power curve in the \citet{Andrews2019} framework of selective publication. Predictive power is the probability of replication for a published finding conditional on a given $p$-value, where replication success is defined using the standard criterion of significance at the 5\% level in the same direction as the original estimate \citep{OpenScience2015, Camerer2016, Camerer2018}. Within a given empirical literature, it can be interpreted as the expected replication probability of a randomly chosen published finding with that $p$-value.\footnote{A $p$-value below 5\% does not, by itself, imply a high probability of replication at the 5\% level. A significant finding may reflect a large true effect, but it may also reflect a favorable sample realization from a small or null effect, in which case it is unlikely to replicate. More generally, for any $p$-value in $(0,1)$, the replication probability can in principle range from 0.025 to arbitrarily close to 1.} We define predictive power under the original design by considering a replication with the same sample size as the original study.\footnote{In practice, replication studies often use larger sample sizes than the original study, which tends to overstate replication rates relative to using the same sample size. Moreover, common power calculations that choose the replication sample size to detect the original estimate with a prespecified power yield expected replication power strictly below that target \citep{Vu2024}.} Importantly, the framework does not require that an actual replication is conducted, or is even feasible. Predictive power is closely related to statistical power and can be applied equally to observational studies, as shown in Section~\ref{section:nonparametric}.




Predictive power has two noteworthy properties. First, as the conditional mean of the replication outcome, it minimizes expected squared prediction error. This provides strong theoretical grounds for using it to estimate replicability at conventional significance thresholds. Second, predictive power is invariant to publication bias. This counterintuitive property arises because publication bias changes the distribution of observed findings \textit{across} $p$-values, but not the distribution of true effects \textit{within} a given published $p$-value, on which predictive power is conditioned. Hence, low predictive power at a given $p$-value cannot be attributed to publication bias.


We estimate the predictive power curve in three empirical literatures with systematic replication evidence: experimental economics \citep{Camerer2016}, psychology \citep{OpenScience2015}, and experimental social science \citep{Camerer2018}. Predictive power is completely determined by the distribution of true effects. We draw directly on estimates from \citet{Vu2024}, whose empirical applications use the `metastudy approach' from \citet{Andrews2019} to estimate this distribution under parametric assumptions.

The advantage of examining literatures with systematic replication evidence is that the estimated predictive power curves can be validated against realized replication outcomes. Predictive power forecasts actual replication outcomes with remarkable accuracy: we cannot reject perfect calibration, and it achieves lower average squared prediction error than prediction markets, despite using far less information. High predictive accuracy provides strong empirical justification for using the measure to evaluate replicability, and also suggests a simple, low-cost alternative to prediction markets for forecasting replication outcomes.


We next turn to the paper's main question: how replicable are findings that narrowly satisfy conventional significance thresholds when replicated using the original sample size? In all three applications, and across all standard conventional significance thresholds, the estimated predictive power curves point to the same conclusion: replication probabilities are generally quite low. For example, in experimental economics---which fares better than psychology and social science---a finding with a $p$-value of 0.05 has a one-in-four chance of being successfully replicated. In other words, a same-sized replication of a just-significant result is three times as likely to be insignificant as significant. At $p=0.01$, the expected replication probability remains relatively low at 0.40. Notably, moving from ``marginal significance'' at the 10\% level to statistical significance at the 5\% level increases expected replication probabilities from 0.19 to 0.25, implying little change in replicability. More generally, achieving high replication probabilities requires far stronger statistical evidence than demanded by conventional thresholds. In experimental economics, achieving a replication probability of at least 0.8 requires $p<0.00006$, a standard met by 16.7\% of findings in the sample. We interpret this as a diagnostic benchmark rather than a recommendation to impose stricter thresholds.

To broaden the applicability of our framework, we develop a nonparametric estimator of predictive power. Its principal advantages are that it imposes no parametric assumptions on the distribution of true effects and can still be used when each study reports multiple correlated $p$-values. Its main drawback is that it cannot be applied to the replication studies due to the small number of observations. Instead, we apply it to a large dataset of experimental and quasi-experimental findings in economics \citep{Brodeur2020}. For randomized control trials, the chance of replicating a finding with $p=0.05$ is 0.34. Quasi-experimental studies using observational data have notably higher replication probabilities at $p=0.05$, ranging from 0.43 to 0.50, although replicability remains low in absolute terms. Across all applications, differences in predictive power can largely be explained by differences in median sample size. Taken together, the parametric and nonparametric estimates reinforce the same conclusion: replicability at conventional significance thresholds is relatively low.



How should statistically significant findings be interpreted in light of these results? Conventional thresholds play a central role in determining whether a finding---often based on a single study---is viewed as evidence of an effect. Yet we find that an empirical result in experimental economics with $p=0.05$ is expected to replicate with probability $\frac{1}{4}$. This does not imply that three-quarters of just-significant findings are false positives. Nor do unsuccessful replications necessarily imply errors in original studies. Rather, low replicability reflects the winner's curse \citep{WinnersCurseQJE}: even if estimates are unbiased before conditioning on their $p$-values, the majority of just-significant estimates overstate their true effects because statistical significance frequently reflects favorable sampling variation around a smaller true effect.

To clarify what these results imply for evidence against the null, Subsection~\ref{subsection:interpreting_replication_outcomes} examines how a successful or unsuccessful replication updates the probability that the true effect is close to zero, beginning with an original finding with $p=0.05$. For experimental economics, we find that a successful replication can substantially strengthen evidence against a near-null effect, reducing its probability from 0.16 to 0.02. By contrast, an unsuccessful replication provides only limited evidence in its favor, increasing the probability from 0.16 to 0.21. This asymmetry arises because a successful replication is highly unlikely when the true effect is near zero, whereas an unsuccessful replication remains relatively likely even when the effect is not close to zero.


Overall, we recommend interpreting statistical significance in a single study as suggestive evidence against the null. Robust scientific conclusions require cumulative evidence across multiple studies. These findings support the longstanding view that replication is fundamental to scientific credibility \citep{Popper1934}, and reinforce more recent calls to strengthen incentives for conducting and publishing replications \citep{Nosek2012,Coffman2015,Brodeur2023}.

\bigskip
\textbf{Related Literature.} The paper contributes to two strands of the metascience literature. First, it contributes to research on replications and statistical power \citep{Cohen1988, Ioannidis2005, Loken2017, Stanley2017, Andrews2019, DellaVigna2022}. It is most closely related to two recent papers. First, \citet{Goodman2022} estimates predictive power in a parametric model without publication bias for medical research and finds low replicability. Second, \citet{Lang2025} estimates 65\% of narrowly rejected null hypotheses in economics are likely to be false rejections. This paper builds on the literature in several respects. First, to our knowledge, it is the first study to validate predictive power against actual replication outcomes. Second, it models publication bias, for which there is substantial empirical evidence across scientific fields \citep{Card1995, Franco2014, Andrews2019}. Finally, it develops a nonparametric estimator of predictive power under weak assumptions.

Second, the paper contributes to the literature on predicting replications \citep{Dreber2015, Altmejd2019, DellaVigna2019, DellaVigna2020, Gordon2021}. Existing approaches use surveys, prediction markets, and machine learning methods to forecast replication outcomes. Our simple model-based forecast uses $p$-values from original studies yet nevertheless outperforms prediction markets. Thus, it offers a simple, low-cost alternative to more resource-intensive forecasting methods.



The remainder of the article is organized as follows. Section \ref{section:theory} derives predictive power. Section \ref{section:empirical} discusses estimation and validation. Section \ref{section:results} presents the parametric results and Section \ref{section:nonparametric} the nonparametric results. Section \ref{section:recommendations} makes recommendations for researchers and replicators. Section~\ref{section:conclusion} concludes.

\section{Theory}\label{section:theory}
This section introduces the framework, derives the predictive power curve, and discusses its properties. Proofs are in Appendix \ref{appendix:proofs}.

\subsection{Setup}\label{subsection:setup}
Consider an empirical literature composed of research findings indexed by $i$. For instance, the literature of interest could be based on field (e.g. labor economics, electoral politics, epidemiology) or methodology (e.g. experimental, quasi-experimental). Research findings are characterized by an (unobserved) true effect, an estimate of the true effect, and a standard error, $(\beta_i, \hat{\beta}_i, \sigma_i)$.\footnote{Following \cite{Andrews2019} and \cite{Wuthrich2022}, we assume for simplicity that the researcher observes the true standard error. Under standard regularity conditions, estimation error in the standard error is typically much smaller than sampling variation in the effect estimate.} We focus on normalized studies, characterized by $(z_i, \hat{z}_i)$, where $z \equiv \beta/\sigma$, $\hat z \equiv \hat \beta/\sigma$. For expositional simplicity, we refer to $z$ as the \textit{true effect} and $\hat z$ as the \textit{estimated effect}.

The data-generating process follows the selective publication model of \citet{Andrews2019} presented in normalized units:

\begin{enumerate}
    \item \textbf{Draw true effect:} Draw $z$ from the distribution of true effects: $z \sim \pi$.
    \item \textbf{Estimate the effect:} Draw an estimate $\hat z$ from a normal distribution centered at the true effect $z$: $\hat z \mid z \sim N(z, 1)$.
    \item \textbf{Publication selection:} Selective publication is modeled by a function $s(\cdot)$, which gives the probability of publication for an estimate $\hat z$. Let $D$ be a Bernoulli random variable equal to 1 if a study is published and 0 otherwise. Then $\mathbbm{P}(D = 1 | \hat z) =   s\left( \hat z \right)$.
\end{enumerate}

We observe $|\hat z|$ for published findings, $D=1$. Focusing on the absolute value reflects common empirical settings in which only the magnitude of the test statistic -- or, equivalently, the two-sided $p$-value -- is reported.

To illustrate the model, consider the quasi-experimental labor economics literature. Papers in this literature address different questions, such as the impact of the minimum wage on employment or the effect of a job-training program on earnings. Each question has an associated true effect $z$, drawn from the distribution $\pi(z)$. For the job-training example, $z$ represents the program's true average treatment effect, normalized by the study's standard error. In the second step, a researcher applies an empirical method -- such as a difference-in-differences design -- to obtain the estimate $\hat z$. Under mild regularity conditions, the resulting estimator is approximately normally distributed with a consistently estimable variance, a standard assumption used in empirical practice. In the final step, the estimate $\hat z$ is published with probability $s(\hat z)$. This captures selective publication; for instance, statistically significant findings may be more likely to be published than null results.

Finally, note that while $\pi(z)$ can be interpreted as a `prior' in a Bayesian hierarchical model, it is not subjective. Instead, it should be interpreted as the \textit{objective} (unobserved) distribution of true effects across studies in the empirical literature of interest.

\subsection{Replication Under Original Design}
\label{subsection:replications_statisticalpower}

Our main objective is to estimate the probability that a finding with a given $p$-value successfully replicates. Formally, consider a replication of an original finding $(z,\hat z)$. To evaluate replicability under the original design, we consider a replication that uses the same sample size as the original study. Thus, the replication estimate satisfies $\hat z_r \mid z \sim N(z,1)$ and is conditionally independent of $\hat z$ given $z$.


Replication success is defined using the standard criterion from large-scale replication studies: statistical significance at the two-sided 5\% level with the same sign as the original estimate \citep{OpenScience2015, Camerer2016, Camerer2018}. Specifically, let the replication indicator be $R =
\mathbbm{1}
\left\{
|\hat z_r| \geq 1.96
\text{ and }
\operatorname{sign}(\hat z_r)
=
\operatorname{sign}(\hat z)
\right\}$, where $1.96$ is the critical value for statistical significance at the two-sided 5\% level. This measure of replication success is closely related to statistical power.\footnote{Conventional one-sided power is the probability that the replication estimate is significant in the direction of the true effect, whereas our measure uses the direction of the original estimate. The two coincide unless the original estimate is statistically significant with the wrong sign, and are therefore likely to be similar when such sign errors are rare.} As with statistical power, it is also well-defined for observational studies for which a direct replication may be infeasible.

We focus on replications using the same sample size as original studies. In practice, replication studies use various sample-size rules, most of which produce larger samples and therefore higher replication rates.\footnote{Appendix \ref{appendix:alternative_rep_sample_size_rules} estimates that the rules used in the experimental economics \citep{Camerer2016}, psychology \citep{OpenScience2015}, and experimental social science \citep{Camerer2018} replication projects raise estimated replication rates by 30--36\% relative to the same-sample-size benchmark. More generally, replication rates depend mechanically on the chosen sample-size rule: for any non-zero effect, statistical power approaches one as sample size increases without bound and approaches the size of the test as sample size approaches zero.} Our benchmark instead evaluates replicability under the original research design, which provides a common basis for comparison across literatures that is independent of replicators' sample-size choices.

\subsection{Predictive Power Curve}\label{subsection:replication_power_curve}
Define \textit{predictive power} as the probability of a successful replication conditional on the observed two-sided $p$-value and publication: $r(p)\equiv \Pr(R=1\mid p,\text{published})$. Tracing $r(p)$ over $p \in (0,1)$ yields the \textit{predictive power curve}. Predictive power is a function of the two-sided $p$-value to reflect the common empirical practice of using two-sided tests. We express predictive power in terms of the $p$-value because statistical significance is typically reported using thresholds such as $p<0.05$ or $p<0.01$. This is without loss of generality because the two-sided $p$-value maps one-to-one to $|\hat z|$.

It is useful to distinguish predictive power from other common definitions of power. \textit{Planned power} is the probability of obtaining a statistically significant result under a prespecified effect size. \textit{Actual power} is this probability under the true effect. Finally, \textit{post hoc power} substitutes the original estimate for the true effect. Because it is a one-to-one transformation of the $p$-value, it provides no additional information about the evidence contained in the original estimate \citep{Hoenig2001}.

The following theorem gives a convenient representation of predictive power:

\begin{theorem}[Predictive Power]\label{theorem:predictive_power}
Let $s(\cdot)$ be symmetric about zero and strictly positive. Then for every $p\in(0,1)$, predictive power is given by
\begin{align}
r(p)
=
\frac{
\int_{\mathbb R}
\bigl[1-\Phi(1.96-z)\bigr]\,
\varphi\!\left(cv(p)-z\right)\,
\widetilde{\pi}(z)\,dz
}{
\int_{\mathbb R}
\varphi\!\left(cv(p)-z\right)\,
\widetilde{\pi}(z)\,dz
}
\label{equation:posterior_replication_power}
\end{align}
where $cv(p)\equiv\Phi^{-1}(1-p/2)$ and $\widetilde{\pi}(z) \equiv \frac{\pi(z)+\pi(-z)}{2}$ is the symmetrized distribution of normalized true effects.
\end{theorem}

Predictive power in \eqref{equation:posterior_replication_power} can be interpreted as the probability that a randomly chosen published finding with a given $p$-value is statistically significant at the 5\% level in the same direction as the original estimate in a replication using the same sample size as the original study. 

The assumptions that $s(\cdot)$ is symmetric about zero and strictly positive are relatively mild. Symmetry requires that publication depend on the magnitude rather than the sign of the test statistic. This is natural in the empirical literatures we examine, which address different research questions and therefore provide no general reason for positive estimates to be more publishable than negative ones. Strict positivity rules out regions of the test-statistic distribution in which publication is impossible. In practice, this is not restrictive, since meta-datasets contain published findings throughout the $p$-value distribution. The use of the symmetrized distribution is necessary because $\pi(\cdot)$ is not identified from absolute test statistics; only its symmetrization $\widetilde{\pi}(\cdot)$ is identified (Theorem~\ref{thm:identification_nonparm}). Nonetheless, the representation in Theorem~\ref{theorem:predictive_power} is exact because predictive power is invariant to replacing $\pi(\cdot)$ with $\widetilde{\pi}(\cdot)$.



Predictive power has three noteworthy features. First, it admits a decision-theoretic interpretation: since $R$ is binary, $r(p)$ is the conditional expectation of replication success given the original $p$-value, and therefore minimizes expected squared prediction error. This provides a theoretical justification for focusing on predictive power to evaluate expected replicability at conventional significance thresholds.

Second, predictive power does not depend on the publication bias function $s(\cdot)$. It is therefore invariant to arbitrary forms of selective publication. The reason for this counterintuitive result is that selection changes the relative frequency of different $p$-values but not the distribution of power conditional on a given $p$-value. In other words, since publication depends only on the observed $p$-value, it treats all studies with that $p$-value identically and therefore leaves the expected replication probability conditional on that $p$-value unchanged. Invariance to publication bias has substantive implications for interpreting $r(p)$. It implies that expected replicability at any fixed $p$-value is identical whether publication bias is highly prevalent or entirely absent. 

Third, predictive power is completely characterized by the symmetrized distribution of normalized true effects, $\widetilde{\pi}(z)$. Consequently, differences in predictive power across literatures arise entirely from differences in $\widetilde{\pi}(z)$. For instance, in a literature containing mostly large true effects, predictive power will be relatively high for a given $p$-value. By contrast, if $\widetilde{\pi}(z)$ places substantial mass near zero, predictive power will be relatively low for the same $p$-value. Thus, the central empirical task is to estimate the symmetrized distribution of normalized true effects, which we turn to next.\footnote{Following \cite{Andrews2019}, we treat $\pi$ throughout as a density. This simplifies the exposition and the proofs for the nonparametric results. The analysis can be extended to discrete or mixed
distributions of true effects without changing the nonparametric estimator, similarly to \cite{faridani2026testingunderpoweredliteratures}.}

\section{Estimation and Validation}\label{section:empirical}
This section estimates the predictive power curve for three empirical literatures and validates its performance against observed replication outcomes. This section focuses on literatures with available replications, which permit direct validation. Section~\ref{section:nonparametric} develops a nonparametric estimator of $r(p)$ and applies it to a larger dataset without replication outcomes. 


\subsection{Data}
We examine three large-scale replication projects that systematically replicated published studies. First, \citet{Camerer2016} replicate experimental economics results from all 18 between-subjects laboratory experiments published in \textit{American Economic Review} and \textit{Quarterly Journal of Economics} between 2011 and 2014. Second, \citet{OpenScience2015} replicate results from 100 psychology findings in 2008 from \textit{Psychological Science}, \textit{Journal of Personality and Social Psychology}, and \textit{Journal of Experimental Psychology: Learning, Memory, and Cognition}. Following \citet{Andrews2019}, we consider a subsample of 73 studies with test statistics that are well-approximated by $z$-statistics. Finally, \citet{Camerer2018} replicate 21 experimental studies in the social sciences published between 2010 and 2015 in \textit{Science} and \textit{Nature}.

\subsection{Parametric Estimation}
Predictive power in equation \eqref{equation:posterior_replication_power} is fully determined by the symmetrized distribution of true effects, $\widetilde{\pi}(z)$. Although publication bias does not enter the expression for $r(p)$, it must still be accounted for when estimating $\widetilde{\pi}(z)$ from published studies. Ignoring selective publication would generally yield inconsistent estimates of $\widetilde{\pi}(z)$, and hence of $r(p)$. To estimate this distribution, we use the metastudy approach of \citet{Andrews2019}, with parameter estimates taken from Table 1 in \citet{Vu2024}. Estimates from \citet{Vu2024} are used because \citet{Andrews2019} estimate the distribution of estimates $\beta$ but not of standard errors, so their estimates do not recover the distribution of $z$ required to calculate predictive power.

The metastudy approach corrects for publication bias and assumes that true effects and standard errors are independent, each following a gamma distribution. It treats findings within each empirical literature as i.i.d. draws from the model described in Subsection~\ref{subsection:setup}. It normalizes true effects to have positive support, so we symmetrize the estimated distribution around zero before calculating predictive power, as required by Theorem~\ref{theorem:predictive_power}. Since our target is the symmetrized $\widetilde{\pi}(z)$, this is a normalization rather than a substantive restriction.\footnote{Equivalently, we can retain the nonnegative distribution of true effects and calculate the probability that the replication estimate is significant with the same sign as the original estimate. This yields identical results.} 

A key advantage of the metastudy approach is that it uses only data from original studies. It can therefore be applied to any suitable meta-sample, regardless of whether replications have been conducted. In our applications, systematic replication evidence is used only to validate the resulting predictive power estimates. The main limitation of the metastudy approach is it assumes that true effects and standard errors are independent. In Section~\ref{section:nonparametric}, we develop a nonparametric estimator that relaxes the independence assumption and the parametric distributional assumptions.

As a robustness check, Figure \ref{figure:posterior_power_curves_repest} in Appendix~\ref{appendix:empirical_results} repeats the analysis using the systematic-replication approach of \citet{Andrews2019}, which uses both original and replication studies and requires weaker identifying assumptions. The resulting predictive power curves are quantitatively similar to those obtained using the metastudy approach.

\subsection{Validation}\label{subsection:validation}
Do the estimated predictive power curves accurately predict whether original studies actually replicate? Answering this question is central to establishing the empirical validity of the measure. If predictive power predicts observed replication outcomes, then it validates the measure and justifies its use in evaluating replicability. Conversely, if it fails to predict observed outcomes, any conclusions drawn from it lack empirical support and should be interpreted with caution.

Unlike the same-sample-size benchmark used in our main analysis, the validation exercise uses the sample-size rules actually implemented in the replication studies. We therefore construct predictive power curves for each field using a modified version of equation \eqref{equation:posterior_replication_power} that replaces original-study standard errors with replication standard errors. This is feasible because, in each application, replication sample sizes were set using prespecified rules.\footnote{Specifically, we replace $1-\Phi(1.96-z)$ with $1-\Phi(1.96-z_r)$, where $z_r=\beta/\sigma_r$ and $\sigma_r$ is determined by the replicators' sample-size rule.} See Lemma \ref{lemma:posterior_power_fpr} in the Appendix for details.

To assess predictive accuracy, we estimate the following linear probability model:
\begin{equation}
R_{il}=\alpha+\beta\hat{R}_{il}+\epsilon_{il},
\label{equation:validation_regression}
\end{equation}
\noindent
where $R_{il}$ is an indicator equal to one if study $i$ in empirical literature $l$ (i.e.\ economics, psychology, or social science) successfully replicates with significance and the same sign as the original estimate, and zero otherwise; and $\hat{R}_{il}$ is its predicted replication probability. Estimation assigns equal weight to each literature. 

We compare predictive power based on the metastudy method with prediction-market forecasts compiled by \citet{Gordon2021}. Before the replications were conducted, academics traded contracts whose final market prices represent aggregate predicted replication probabilities. Because prediction markets were conducted for only a subset of replicated findings, the validation exercise uses the common sample of 70 studies for which both predictions are available.

Figure \ref{figure:model_fit} presents the results graphically using binscatter plots with fitted lines. Visually, both predictors are highly predictive of realized replication outcomes. Remarkably, the fitted line for predictive power lies almost exactly on the 45-degree benchmark of perfect calibration, under which predicted replication probabilities equal realized replication rates. By contrast, the fitted prediction-market line lies below the benchmark, indicating that prediction markets are overly optimistic about replication success. Formally, a joint test of perfect calibration ($\alpha=0$, $\beta=1$) does not reject the null for predictive power ($p=0.762$), but marginally rejects it for prediction markets at the 10\% level ($p=0.084$).

\begin{figure} [t]
	\centering
 	\caption{Replication Prediction Accuracy}
	\label{figure:model_fit}
\includegraphics[width=0.625\textwidth]{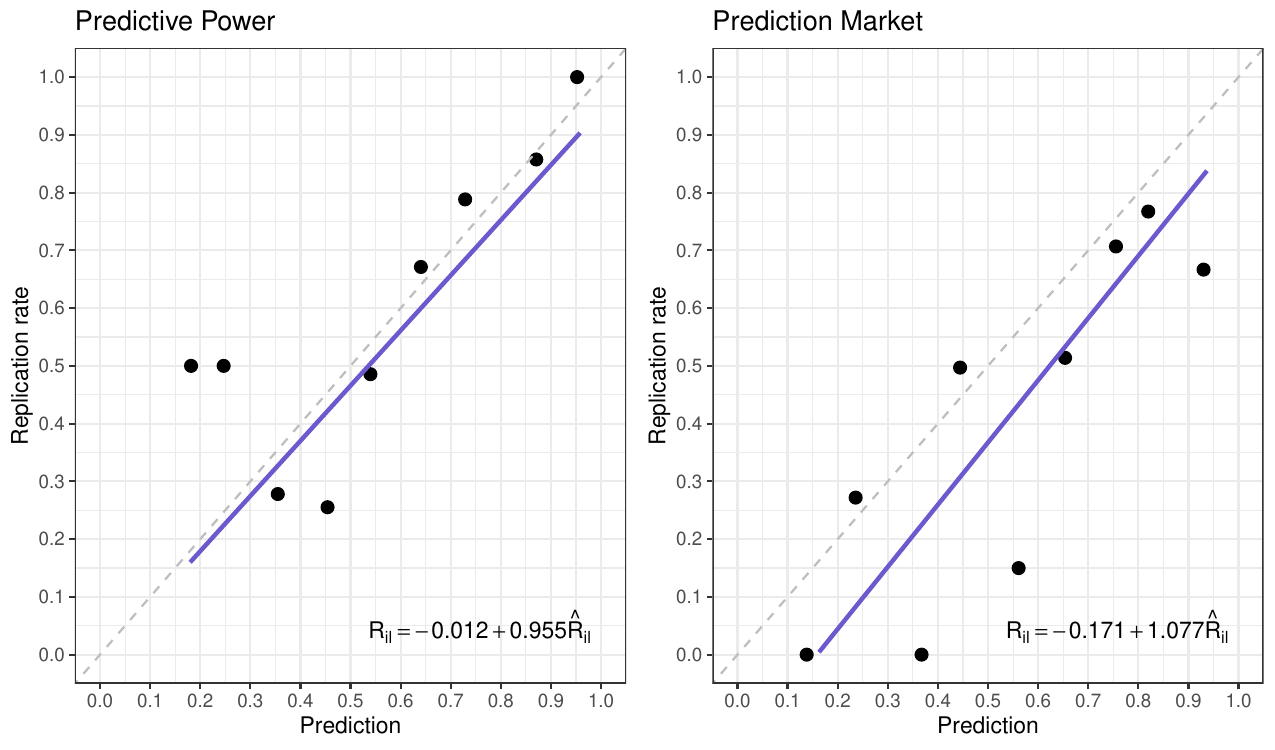} 
\caption*{\textit{Notes}: The figure presents binscatter plots corresponding to equation \eqref{equation:validation_regression} for predictions from predictive power and prediction markets. Solid lines show fitted values from estimating equation \eqref{equation:validation_regression}, and dashed lines show the 45-degree benchmark. Perfect calibration is assessed using a joint test of $\alpha=0$ and $\beta=1$. The corresponding $p$-values are 0.762 for predictive power and 0.084 for prediction markets. Prediction-market forecasts are from \citet{Gordon2021}.}
\end{figure} 

Table \ref{table:predictive_accuracy} complements the calibration regressions by comparing predictors using average squared prediction error, also known as the Brier score.\footnote{The Brier score is strictly proper: a forecaster minimizes expected loss by reporting their true subjective probability. This makes it suitable for comparing probability forecasts because it rewards honest, well-calibrated probabilities rather than strategic overstatement or understatement.} Calibration in Figure~\ref{figure:model_fit} assesses whether predicted probabilities match realized replication rates on average, whereas the Brier score also rewards forecasts that distinguish between studies with different replication prospects. For example, assigning every study a probability of 0.30 is perfectly calibrated if 30\% replicate, but provides no information about which studies are more likely to replicate. The Brier score addresses this by penalizing prediction errors at the study level, thereby favoring forecasts that are both well calibrated and informative.

In the weighted pooled sample, predictive power slightly outperforms prediction markets, with Brier scores of 0.216 and 0.220, respectively. This is quite striking given its limited information set. It relies only on test statistics from original studies, combined with the simple model outlined in Section \ref{subsection:setup}. Since the distribution of true effects is estimated from the full distribution of test statistics in the literature, each predicted replication probability places an individual result in the context of the broader body of evidence. By contrast, prediction markets aggregate the judgments of many expert scientists, whose information sets may include detailed knowledge of the specific research question, broader field expertise, perceptions of author quality, and evidence from earlier replication studies. That a parsimonious statistical model can slightly outperform this aggregation of expert judgments suggests that much of the signal relevant for replication is already embedded in the reported statistical evidence itself.

\begin{table}[!t] \centering 
 \scriptsize
\caption{Replication Forecasting Error}
  \label{table:predictive_accuracy} 
\begin{tabular}{@{\extracolsep{5pt}} lcc} 
\\[-1.8ex]\hline 
\hline \\[-1.8ex] 
 & Predictive Power & Prediction Market \\ 
\hline \\[-1.8ex] 
Economics  & $0.188$ & $0.226$ \\ 
Psychology  & $0.243$ & $0.243$ \\ 
Social Science  & $0.217$ & $0.190$ \\ 
All (Weighted) & $0.216$ & $0.220$ \\ 
\hline \\[-1.8ex] 
\end{tabular} 
\caption*{\textit{Notes}: This table reports Brier scores for predictive power and prediction-market forecasts. For finding $i$, the Brier score is $(R_i-\hat R_i)^2$, averaged across findings. Lower scores indicate greater predictive accuracy.}
\end{table} 

Overall, predictive accuracy supports using predictive power to evaluate replicability and offers a simple, low-cost alternative to prediction markets.

\section{Results}\label{section:results}

\subsection{Replicability at Conventional Significance Thresholds} 
We turn now to the main question of interest: how likely are findings that just satisfy conventional significance thresholds to replicate when using the same sample size as the original study? Figure~\ref{figure:posterior_power_curves} plots the predictive power curve for experimental economics, psychology, and social science. Evaluating the curve at the 5\% significance threshold, we find that replicability ranges between 0.1 and 0.25 across fields. Thus, a randomly chosen finding that just meets the most commonly used threshold for statistical significance has, at most, about a one-in-four chance of successful replication. At the more stringent 1\% level, replication probabilities remain relatively low, ranging between 0.23 and 0.40 across fields. Thus, even at the most stringent threshold typically used in empirical research, the chance of replication remains less than one-half across fields.

\begin{figure} [!t]
\centering
\caption{Predictive Power Curves, by Replication Project}
\label{figure:posterior_power_curves}
\includegraphics[width=0.64\textwidth]{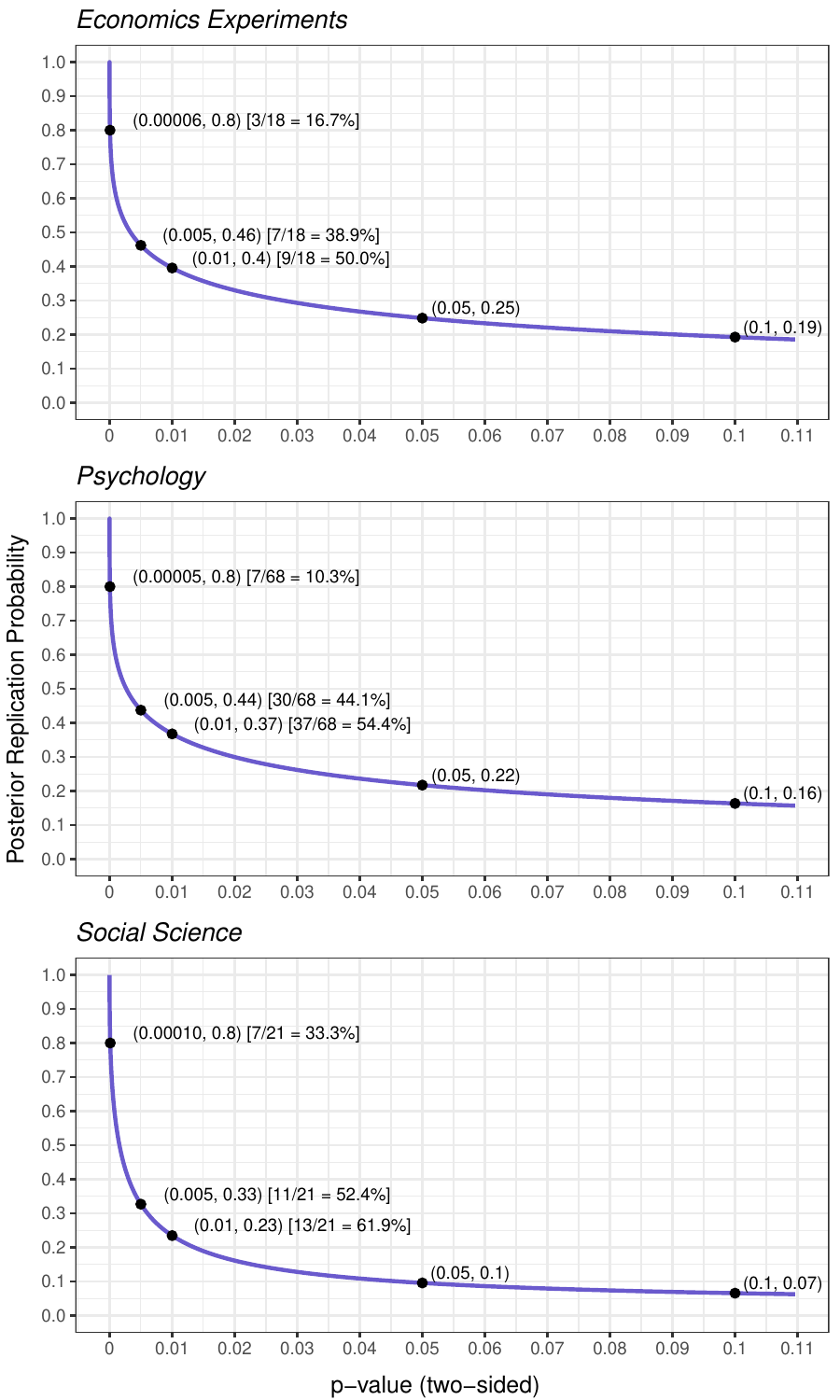}
\caption*{\textit{Notes}: The figure plots estimated predictive power curves for experimental economics, psychology, and experimental social science, using the metastudy approach to estimate $\tilde{\pi}(z)$. Predictive power is the probability that a replication with the same sample size as the original study produces an estimate that is statistically significant at the two-sided 5\% level and has the same sign as the original estimate, conditional on the original study's two-sided $p$-value. Square brackets report the percentage of original studies with $p$-values at or below the indicated threshold. Underlying data are from the replication projects of \citet{Camerer2016}, \citet{OpenScience2015}, and \citet{Camerer2018}, respectively.}
\end{figure}

These results come from a model that abstracts from several commonly proposed causes of low replicability, including $p$-hacking, researcher manipulation, measurement error, and heterogeneity between original and replication studies \citep{Ioannidis2005, Brodeur2016, Loken2017, Wuthrich2022}. Moreover, as discussed in Subsection~\ref{subsection:replication_power_curve}, predictive power is invariant to arbitrary forms of publication bias. The model nevertheless accurately predicts observed replication outcomes (Figure~\ref{figure:model_fit} and Table~\ref{table:predictive_accuracy}), suggesting that these factors are not necessary for explaining the low rates of replicability observed in practice.

The mechanism driving these results is instead the substantial mass that the estimated distributions of true effects place on small effects, leaving original studies with low statistical power at conventional significance thresholds. Appendix Figure~\ref{figure:power_bins} shows the distribution of predictive power conditional on $p=0.05$ for each application. A substantial share of just-significant findings have extremely low replicability: the share with replication probabilities below 5\% is 20\% in economics, 30\% in psychology, and 77\% in the social sciences. Thus, very low replication probabilities are relatively common across all fields and particularly prevalent in the social sciences.

These results imply that many just-significant findings have true effects substantially smaller than their original estimates. Figure~\ref{figure:posterior_z} in Appendix~\ref{appendix:empirical_results} plots the distributions of the true effect $z$ conditional on $p=0.05$ for all three applications. The distributions place substantial mass below the observed estimate of $\hat z=1.96$, although their shapes differ across fields: the distribution is centered at a positive but considerably smaller effect in economics, while it is concentrated closer to zero in psychology and especially in social science. This bias reflects the ``winner's curse" \citep{WinnersCurseQJE}: even though the original estimators are unbiased before conditioning on the $p$-value, just-significant estimates often combine relatively small true effects with favorable sampling variation, causing them to overstate the underlying effects. Consequently, many findings with $p=0.05$ are also likely to fail replication criteria based on similarity in effect magnitude, such as the relative effect size, $\hat z_r/\hat z$.


\subsection{Interpreting Replication Outcomes}\label{subsection:interpreting_replication_outcomes}
How should replication outcomes be interpreted in light of these results? We begin by clarifying what the results do not imply. Low replicability at conventional significance thresholds does not mean that most just-significant findings are `false'. Nor does an unsuccessful replication refute the original finding or necessarily indicate errors in the original study. Instead, replication outcomes should be interpreted as additional evidence about the null hypothesis that contributes to the cumulative evidence on an empirical claim. 

To formalize this idea, consider a researcher interested in whether the true effect is `close to zero'. This framing meets practitioners ``where they are,'' given the ubiquity of null-hypothesis significance testing in applied research.\footnote{For proposals to move beyond dichotomous significance testing, see \citet{McShane2019} and \citet{Amrhein2019}.} Define a \textit{near-null effect} as a true effect whose magnitude is small enough that a same-sized replication has less than a 5\% probability of achieving 5\% two-sided significance in the direction of the true effect.\footnote{We focus on near-null effects because the estimated distribution of true effects is continuous and therefore assigns zero probability to an effect of exactly zero.} Thus, a near-null corresponds to $\{|z|<0.315\}$.\footnote{For $z>0$, the probability that a same-sized replication is significant in the direction of the true effect is $1-\Phi(1.96-z)$. Equating this probability to 0.05 gives $z=1.96-\Phi^{-1}(0.95)\approx0.315$; symmetry gives the corresponding cutoff for $z<0$. Because the definition uses the normalized effect $z$, rather than the underlying treatment effect $\beta$, the cutoff does not represent a common substantive effect size across studies and should be viewed as an operational definition of a near-null effect.}


Let $q_0(p)\equiv\Pr(H_0\mid p)$ denote the probability of a near null conditional on the original $p$-value. Recall that $r(p)\equiv\Pr(R=1\mid p)$ is the probability of a successful replication, and define $r_0(p)\equiv\Pr(R=1\mid p,H_0)$ as the corresponding probability under a near null. By construction, $r_0(p)<0.05$. Then it follows that
\begin{equation}\label{equation:updating}
\begin{aligned}
\Pr(H_0\mid p,D=1,R=1)
&=
q_0(p)
\left(
\frac{r_0(p)}{r(p)}
\right)
\\
\Pr(H_0\mid p,D=1,R=0)
&=
q_0(p)
\left(
\frac{1-r_0(p)}{1-r(p)}
\right)
\end{aligned}
\end{equation}

The first term expresses the probability of a near null after a successful replication and the second after an unsuccessful replication. In both expressions, the first term is the probability of a near null based on the original evidence, and the second is the updating factor implied by the replication outcome.

\begin{figure}[!t]
\centering
\caption{Probability of Near Null Before and After Replication}
\label{figure:effective_null_updates}
\includegraphics[width=0.64\textwidth]{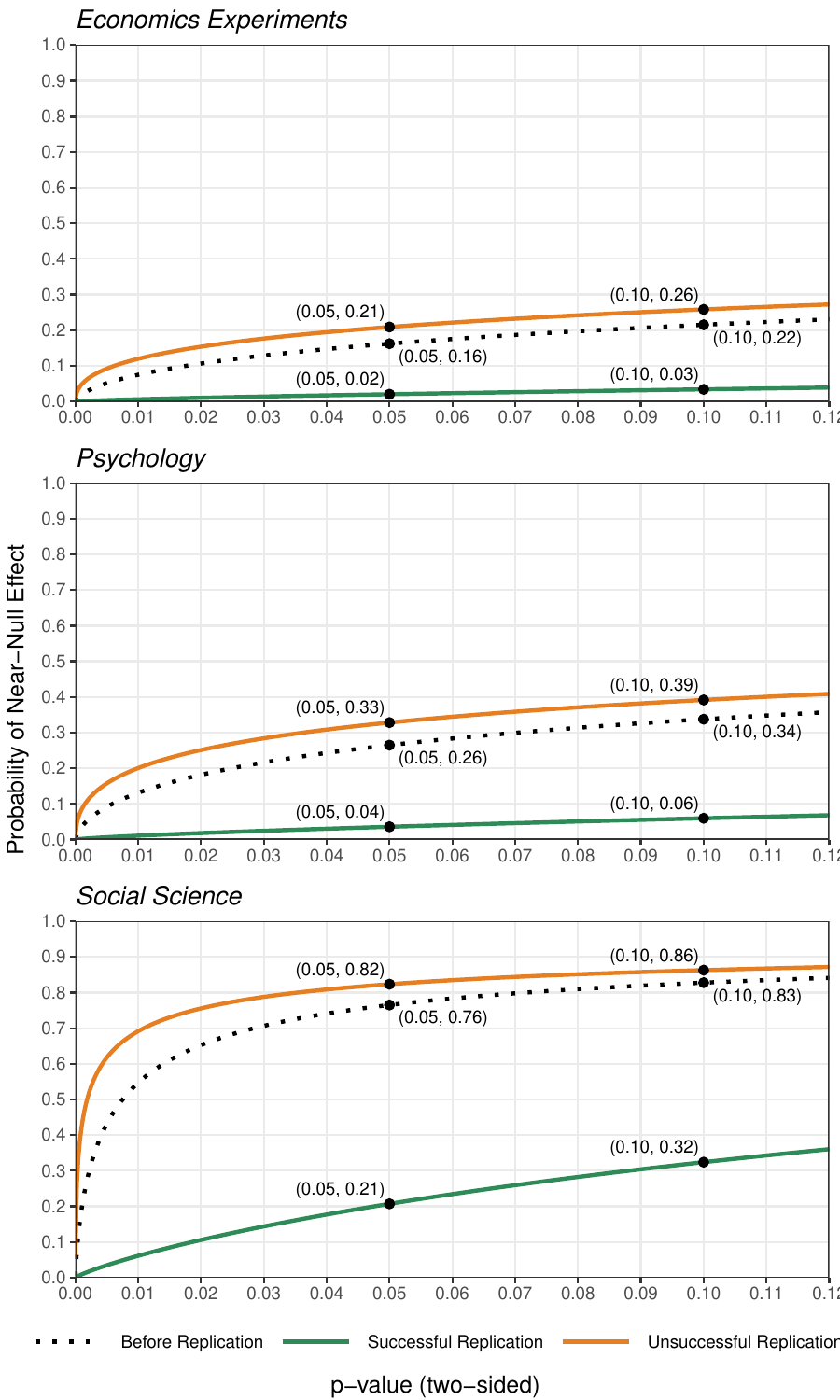}
\caption*{\textit{Notes}: The figure plots the probability of a near null before replication, after a successful replication, and after an unsuccessful replication, as a function of the original two-sided $p$-value. A near null is defined by $\{|z|<0.315\}$, corresponding to a true effect whose magnitude implies less than a 5\% probability that a replication achieves 5\% two-sided significance in the direction of the true effect. A successful replication is statistically significant at the 5\% level and has the same sign as the original estimate. Replications are assumed to use the same sample size as the original study. Underlying data are from \citet{Camerer2016}, \citet{OpenScience2015}, and \citet{Camerer2018}, respectively.}
\end{figure}
Figure~\ref{figure:effective_null_updates} illustrates the updating in \eqref{equation:updating}, plotting the probability of a near null effect before replication (dotted line), after a successful replication (green line), and after an unsuccessful replication (orange line). Most strikingly, for most $p$-values, a successful replication sharply lowers the probability of a near null, whereas an unsuccessful replication raises it only modestly.

For example, in experimental economics, a finding with $p=0.05$ has a 0.16 probability of being a near null before replication. A successful replication reduces this probability sharply to 0.02, whereas an unsuccessful replication increases it only modestly to 0.21. This asymmetry arises because near nulls have very low replication probabilities, ranging from 0.025 to 0.05. A successful replication is therefore highly unlikely under $H_0$ and produces a large downward update in \eqref{equation:updating}, whereas an unsuccessful replication remains relatively likely even when the effect is not close to zero, and therefore produces only a modest upward update. This asymmetry also appears in psychology and experimental social science, although near nulls are generally more likely in these fields for the same original $p$-value and replication outcome.

This asymmetry has implications for interpreting replication outcomes in practice: a successful replication can substantially strengthen evidence against a near null, whereas an unsuccessful replication generally provides only limited evidence in its favor. Thus, replication outcomes are best interpreted as cumulative evidence rather than binary confirmation or refutation.

\subsection{Additional Results}
We highlight two additional implications of the predictive power curves. The first concerns how strong the statistical evidence in the original study must be to achieve a high probability of replication. For any replication probability target $T$, inverting the predictive power curve yields the $p$-value cutoff $p^*$ satisfying $r(p^*)=T$. Lemma~\ref{lemma:uniqueness} proves that, under mild regularity conditions, this cutoff exists and is unique. Figure~\ref{figure:posterior_power_curves} shows that achieving predictive power of at least 80\% requires $p$-value cutoffs of 0.00006, 0.00005, and 0.0001 in economics, psychology, and the social sciences, respectively. Between 10\% and 33\% of findings in the respective samples meet these literature-specific benchmarks. Notably, these cutoffs are several orders of magnitude more stringent than conventional significance thresholds, and far below even the 0.005 threshold proposed in \citet{Benjamin2018}. Nevertheless, we do not interpret these cutoffs as a recommendation to mechanically adopt more stringent thresholds. Although stricter thresholds would increase the average replicability of findings that meet them, they might also discard informative evidence, disadvantage settings where very large samples are infeasible, and intensify specification searching.

The second additional result concerns the convexity of the predictive power curves. In current empirical practice, moving from a $p$-value of 0.10 to 0.05 is commonly treated as a material gain in credibility, marking a transition from ``marginal significance'' to genuine ``statistical significance.'' However, Figure~\ref{figure:posterior_power_curves} suggests that this transition does little to alter the underlying replicability of a finding: expected replication probabilities rise by only 3--6 percentage points, and remain at low levels. This implies that moving between conventional significance thresholds need not translate into meaningful improvements in replicability. 

This observation also has implications for forms of $p$-hacking that misreport $p$-values just above a conventional threshold as falling just below it. Unlike selective reporting, which changes whether a result is observed but not its reported $p$-value, this form of $p$-hacking changes the reported $p$-value itself. Consider a finding with a true $p$-value of 0.06 that is misreported as $p=0.049$ in order to claim significance; its expected replicability remains $r(0.06)$. Figure~\ref{figure:posterior_power_curves} shows that $r(0.06)$ and $r(0.049)$ are nearly identical across all fields. Thus, forms of $p$-hacking that nudge reported $p$-values across the 5\% significance threshold are unlikely to account for much of the low replicability observed in practice.

\section{Nonparametric Predictive Power Curve}\label{section:nonparametric}

Thus far, the empirical analysis has imposed parametric assumptions about the distribution of true effects. In this subsection, we develop a nonparametric estimator of $r(p)$ that imposes no functional-form restrictions on $\tilde \pi(z)$. The nonparametric estimator requires large meta-samples that include many reported $p$-values and cannot be applied to the relatively small replication datasets analyzed above. Given these data requirements, we apply the estimator to the large dataset of experimental and quasi-experimental economics findings compiled by \citet{Brodeur2020}.

Nonparametric identification of $r(p)$ is established formally in Appendix \ref{app:nonparametrics}. Discontinuities in the density of published $t$-ratios identify changes in publication probability across significance thresholds. After adjusting for these changes, recovering the distribution of true effects from the distribution of estimated $t$-ratios is a standard deconvolution problem with normally distributed errors.\footnote{See e.g. \cite{Carroll,Fan_global,CarrascoFlorens,faridani2026testingunderpoweredliteratures}.} Below, we focus on the principal challenge of constructing a consistent estimator of $r(p)$. Proofs are in Appendix \ref{app:nonparametrics}.

\subsection{Nonparametric Estimation}

Suppose we observe a sample of absolute values of $n$ published $t$-ratios, $\{|\hat z_i|\}_{i=1}^n$.  Our aim is to estimate the symmetrized density $\tilde{\pi}(z) \equiv \frac{\pi(z)+\pi(-z)}{2}$, which determines the target parameter $r(p)$ and is identified from the distribution of $|\hat z|$. Following Example 1 of \citet{CarrascoFlorens} and \citet{faridani2026testingunderpoweredliteratures}, we can express  $\tilde\pi(z)$ as a linear combination of known polynomial basis functions, $\chi_j$, with unknown coefficients $b_j$:
\begin{equation}\label{eq:pi_series}
    \tilde{\pi}(z) = \sum_{j=0}^\infty b_j\chi_j(cv(p)-z),
\end{equation}
\noindent where $\chi_j(z) \equiv \He{j}{z/\sqrt{2}}/\sqrt{j!}$ is the normalized $j$th probabilists' Hermite polynomial. 

\begin{remark}
    Estimating the full density $\tilde \pi$ is severely ill-posed because many distributions of true effects can generate nearly indistinguishable distributions of observed $t$-ratios \citep{Carroll,Fan_global}. Distinguishing among them relies on higher-order coefficients $b_j$, which capture fine features of the data and are highly sensitive to sampling error. Despite this difficulty, the key insight underlying our estimator is that higher-order coefficients have diminishing influence on $r(p)$ because it averages over $\tilde{\pi}(z)$ and is therefore relatively insensitive to its finer features. Consequently, imprecise estimates of $\tilde{\pi}(z)$ can yield accurate estimates of $r(p)$. A similar principle appears in \citet{faridani2026testingunderpoweredliteratures}, but targeting $r(p)$ yields a new estimator and a distinct convergence rate.
\end{remark}

In practice, we truncate the infinite expansion and estimate only the first $J_n$ coefficients. The cutoff $J_n$ grows logarithmically with the sample size, allowing the estimator to capture finer features of $\tilde{\pi}(z)$ as more data become available. Formally, this amounts to regularizing the deconvolution problem using a spectral cutoff based on the singular value decomposition from \cite{CarrascoFlorens}.\footnote{See Example 1 of \cite{CarrascoFlorens} for the singular value decomposition. The Gibbs fluctuations incurred by spectral cutoff do not prevent consistency because they are integrated away when we calculate $r(p)$ from $\tilde{\pi}(z)$.} Substituting this truncated expansion into Equation~(\ref{equation:posterior_replication_power}) and simplifying yields the following estimator:
\begin{align}\label{equation:np_estimator}
   \widehat{r}^{np}(p)
   \equiv
   1-\frac{\sum_{j=0}^{J_n}a_j\widehat{b}_j}
   {\sum_{j=0}^{J_n}d_j\widehat{b}_j}
\end{align}

\noindent where $\widehat{b}_j$ estimates the coefficient $b_j$ in \eqref{eq:pi_series} times a common constant (that cancels in the ratio). The known constants $a_j$ and $d_j$ capture the contribution of the $j$th basis function to the numerator and denominator, respectively, and can be computed numerically.\footnote{\raggedright More specifically,
$a_j \equiv \int_{\mathbb{R}}
\varphi(u)e^{-u^2/2}
\Phi\left(c-cv(p)+\sqrt{2}u\right)\psi_j(u)\,du$,
and
$d_j \equiv \int_{\mathbb{R}}
\varphi(u)e^{-u^2/2}\psi_j(u)\,du$.\par}

To estimate the coefficients $b_j$ consistently, we must account for selective publication of the observed $t$-ratios. As before, let $D=1$ indicate that a finding is published, and suppose its publication probability depends on the observed $t$-ratio through the even selection function $s_{\theta_0}$:
\begin{equation}\label{eq:w_definition_main}
    \Pr(D=1\mid \hat z )=s_{\theta_0}(\hat z)
\end{equation}
\noindent where the functional form of $s_\theta$ is known but the parameter vector $\theta_0\in \Theta$ is unknown.

The properties of the Hermite basis allow each coefficient $b_j$ to be expressed as a moment of the unselected distribution of estimated $t$-ratios. Given a consistent estimator $\widehat{\theta}_n$ of the publication-bias parameters, we follow \citet{faridani2026testingunderpoweredliteratures} and estimate this moment by weighting each published $t$-ratio by the inverse of its estimated publication probability:
\begin{align}\label{equation:np_coefficients}
    \widehat{b}_j
    &\equiv
    \frac{2^{j/2}}{n}\sum_{i=1}^n
    \frac{
         \psi_j\left(cv(p)-|\hat z_i|\right)
        \varphi\left(cv(p)-|\hat z_i|\right) +\psi_j\left(cv(p)+|\hat z_i|\right)
        \varphi\left(cv(p)+|\hat z_i|\right) 
    }{
        s_{\widehat{\theta}_n}\left(\hat z_i\right)
    }
\end{align}
where $\varphi(\cdot)$ is the standard normal density, $\psi_j(z)\equiv \He{j}{z}/\sqrt{j!}$ and $cv(p)\equiv\Phi^{-1}(1-p/2)$ is the standard normal critical value corresponding to the two-sided $p$-value.

Next we formalize our assumptions about the sample. In particular, we allow test statistics to be dependent within an article, but not across articles.
\begin{assumption}\label{assum:sample}
    The following all hold:
    \begin{enumerate}
        \item Let $|\hat z_1|,\ldots,|\hat z_n|$ be an identically distributed sample from the distribution of $|\hat z|$ conditional on publication, i.e. $D=1$.
        \item Sets of $\hat z_i$ drawn from different articles are independent.
        \item There is a universal constant $L\geq 1$ such that each article contributes no more than $L$ observations to the sample. 
    \end{enumerate}
\end{assumption}

Theorem \ref{thm:nonparametric_consistency_simplified} shows that, if the publication-bias parameters $\theta_0$ are known or consistently estimated by $\widehat{\theta}_n$, then the estimator $\widehat{r}^{np}(p)$ defined in \eqref{equation:np_estimator}, with coefficients given by \eqref{equation:np_coefficients}, is consistent for $r(p)$.\footnote{This result is a special case of the more general Theorem \ref{thm:nonparametric_consistency} in Appendix \ref{app:nonparametrics}, which allows the replication and original sample sizes to differ.}

\begin{theorem}\label{thm:nonparametric_consistency_simplified}
    Let the sample $|\hat z_1|,\ldots,|\hat z_n|$ satisfy Assumption \ref{assum:sample}. Suppose that the density $\pi(z)$ exists and has bounded height. Assume that $s_\theta$ is even and $\inf_{\theta\in\Theta,\,t\in\mathbb R}s_\theta(t)>0.$ Let $J_n=\left\lceil \frac{2}{3\log 2}\log(An)\right\rceil$ for any $A>0$. Assume further that an estimator $\widehat{\theta}_n$ satisfies: $$
        \sup_{t\in\mathbb{R}}
        n^{1/3}\left|
        \frac{1}{s_{\widehat{\theta}_n}(t)}
        -\frac{1}{s_{\theta_0}(t)}
        \right|
       =\mathcal{O}_p(1)$$ Then, for every $p\in(0,1)$,
    \[
        \widehat{r}^{np}(p)\xrightarrow{p}r(p)
    \]
\end{theorem}

We complete the estimation procedure by specifying and estimating the publication-bias function. To allow for flexible forms of publication bias, we model the publication probability as a step function with $M$ known critical values $\{v_m\}_{m=1}^M$ where the publication probability of the most-significant $\hat z$ is normalized to one:
\begin{align}
    s_{\theta}(\hat z)
    &= \theta_{(1)}\mathbf{1}\left\{|\hat z| \leq v_1\right\}
    + \sum_{m=2}^{M}\theta_{(m)}
    \mathbf{1}\left\{|\hat z| \in (v_{m-1},v_{m}]\right\}
    + 
    \mathbf{1}\left\{|\hat z|>v_M\right\}
    \label{eq:stepfunction}
\end{align}

Given the known step locations, we estimate the publication-bias parameters using the caliper ratio estimator of \citet{faridani2026testingunderpoweredliteratures}. The estimator compares the numbers of published $t$-ratios within narrow intervals on either side of each critical value. Discontinuities in these counts identify changes in publication probability, from which the step heights can be recovered. The resulting estimates are uniformly consistent for $s_{\theta}$ and can be used in Theorem~\ref{thm:nonparametric_consistency_simplified}. For further details, see \citet{faridani2026testingunderpoweredliteratures}.

\subsection{Application}
We apply the nonparametric estimator to the dataset in \cite{Brodeur2020}, which consists of 21,740  $t$-ratios testing main hypotheses in 684 articles published in top economics journals. Each $t$-ratio corresponds to the test of a causal inference hypothesis where the study methodology is either a randomized controlled trial (RCT), Regression Discontinuity Design (RDD), Difference-in-Differences (DID), or Instrumental Variables (IV). Papers can contain multiple estimates. Following \cite{Kranz2022}, we de-round the $t$-ratios before estimation because rounding reported coefficients and standard errors can generate artificial bunching around conventional significance thresholds. The publication bias function, $s_{\theta}(\cdot)$, is specified to have steps at the three conventional two-sided critical values $1.645, 1.96,$ and $2.576$. Estimates for the publication probability function are presented in Appendix~\ref{appendix:publication_probability_ratios}.

\begin{figure}[!t]
    \centering
    \caption{Nonparametric Predictive Power Curves, by Research Design in Economics} \label{fig:brodeur-stacked-nonparametric}
    
 \includegraphics[width=\textwidth]{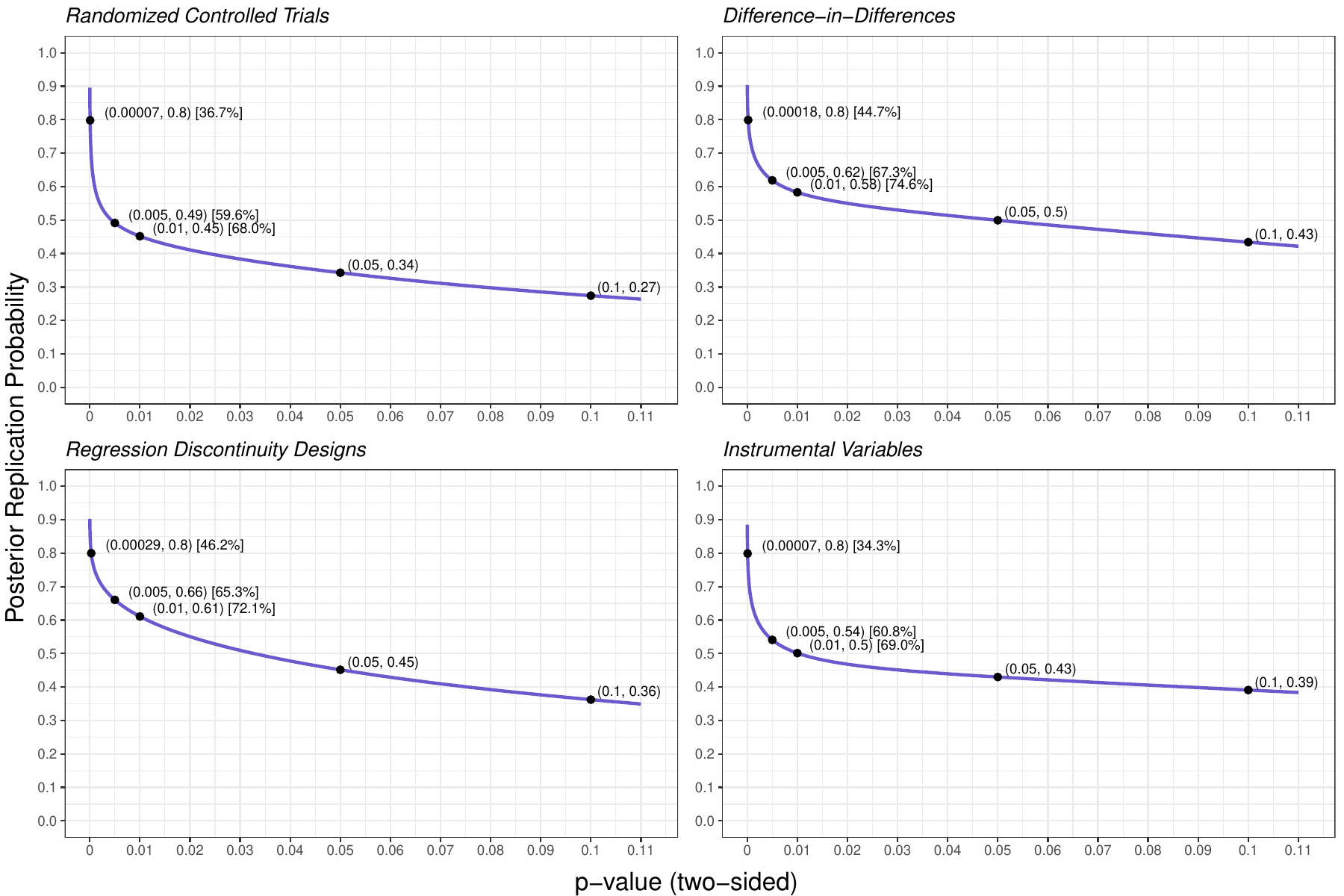}
\caption*{\textit{Notes}: The figure plots nonparametric estimates of the predictive power curve by research design in economics. Predictive power is evaluated for a replication with the same sample size as the original study. Underlying data are from \citet{Brodeur2020}.}
\end{figure}

To make the replication exercise concrete for observational studies, consider an RDD using the Current Population Survey (CPS) to estimate the effect of Medicare eligibility at age 65 on health-insurance coverage. Replication is a hypothetical repeated-sampling exercise: imagine drawing a new independent sample of the same size from the same population and policy environment and applying the identical RDD specification. Replication is successful if the resulting estimate is significant at the 5\% level and has the same sign as the original estimate.

Figure \ref{fig:brodeur-stacked-nonparametric} presents the nonparametric predictive power curves for the four research designs. An RCT finding with a $p$-value of 0.05 has an estimated replication probability of 0.34, modestly exceeding the parametric estimate of 0.25 for experimental economics. Predictive power at $p=0.05$ is notably higher for the observational designs, ranging between 0.43 and 0.50, although, in absolute terms, replication probabilities remain relatively low. Similar to the parametric analysis in Figure~\ref{figure:posterior_power_curves}, extremely small $p$-values are required to achieve predictive power of at least 80\%.

Differences in replicability across the seven applications are closely related to differences in sample size. Figure~\ref{figure:median_sample_size_r0.05} plots median sample size against predictive power at $p=0.05$. It shows that predictive power is approximately linear in the logarithm of median sample size, with the applications lying close to the fitted line. Higher predictive power in observational studies is consistent with their much larger samples. Similarly, economics RCTs—which include relatively large field experiments—have higher replicability than smaller laboratory studies in experimental economics. These patterns are consistent with replicability being determined primarily by the distribution of normalized true effects, $\tilde \pi(z)$, as emphasized in Section~\ref{section:results}. For a given underlying effect, a larger sample reduces the standard error and increases the normalized true effect, shifting $\tilde \pi(z)$ away from zero and increasing predictive power.

These results offer a different perspective on the relationship between research practices and replicability. A literature with more selective reporting and $p$-hacking may nevertheless be more replicable. For example, \citet{Brodeur2020} use discontinuities at conventional significance thresholds to argue that IV studies—and, to a lesser extent, DID studies—exhibit substantially more $p$-hacking than RCTs. Yet Figure~\ref{figure:median_sample_size_r0.05} shows that just-significant RCT findings have much lower expected replicability than IV and DID findings, consistent with their smaller sample sizes.

This is because replicability depends primarily on the distribution of normalized true effects, which is strongly influenced by sample size. The distinction matters because low replicability is frequently attributed to $p$-hacking and publication bias \citep{OpenScience2015, Camerer2016, Camerer2018}. For example, in a 2016 \textit{Nature} survey, 90\% of researchers identified ``selective reporting'' as contributing to irreproducible research, more than any other factor \citep{Baker2016}. Against this background, our results suggest greater caution in attributing low replicability primarily to selective reporting and $p$-hacking.

\begin{figure} [!t]
\centering
\caption{Median Sample Size and Expected Replicability at $p=0.05$}
\label{figure:median_sample_size_r0.05}
\includegraphics[width=0.5\textwidth]{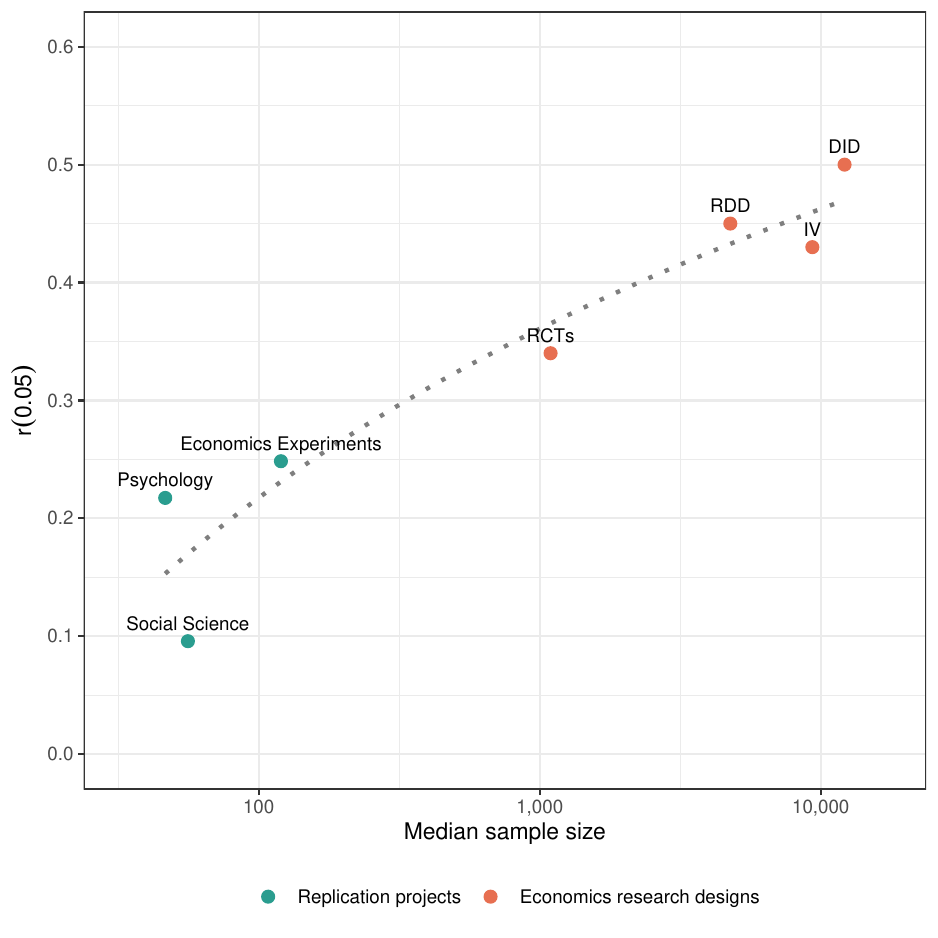}
\caption*{\textit{Notes}: The figure plots the expected replication probability at an original two-sided $p$-value of 0.05 against the median sample size for each application. Replication project data are from \citet{Camerer2016}, \citet{OpenScience2015}, and \citet{Camerer2018}. For RCTs, IV, RDD, and DiD, data are from \citet{Brodeur2020}, with median sample sizes calculated from the nonmissing values. The dotted line plots the OLS fitted values from $r_i(0.05)=\alpha+\beta\log_{10}(\operatorname{Median}N_i)+\varepsilon_i$.}
\end{figure}

Overall, the nonparametric analysis extends the earlier parametric results to a broader range of research designs while allowing for a flexible distribution of true effects. Expected replicability at conventional significance thresholds is slightly higher for observational designs than for experiments, but remains low in absolute terms across all applications.\footnote{The additional analyses in Figure \ref{figure:effective_null_updates} are not feasible with the nonparametric estimator because they require estimation of $\tilde{\pi}$ itself. While the estimator of $r(p)$ converges relatively quickly, the nonparametric estimator of $\tilde{\pi}$ converges slowly and is not guaranteed to yield a valid density.}

\section{Recommendations}\label{section:recommendations}
This section offers three recommendations for researchers and replicators. Recommendations are broadly applicable to replications evaluated using either binary significance criteria or effect sizes.

\subsection{Cumulative Evidence}\label{subsection:cumulative_evidence}
\begin{recommendation}
Treat evidence as cumulative, not binary. A single significant finding provides only suggestive evidence against the null, while evidential strength accumulates through replication. Successful replications can provide strong evidence against the null, whereas insignificant replications do not, by themselves, refute the original finding or imply flaws in the study.
\end{recommendation}


A significant finding provides only partial evidence against the null hypothesis. It does not establish that the null is false. Similarly, a replication does not deliver a definitive conclusion about the original claim, confirming it if ``successful'' and refuting it otherwise. 

For interpreting replication outcomes, the analysis in Subsection~\ref{subsection:interpreting_replication_outcomes} reveals an important asymmetry: a successful replication can substantially strengthen evidence against a near null result, whereas an unsuccessful replication generally provides only limited evidence in its favor. In this way, replications build the cumulative body of evidence required to distinguish persistent empirical relationships from isolated results driven by sampling variation, leading to more reliable scientific inferences. The ultimatum-game literature provides an instructive example, where a large number of replications have established robust empirical regularities \citep{Coffman2015}. Treating scientific evidence as cumulative aligns with arguments developed by \citet{Popper1934} in \textit{The Logic of Scientific Discovery}: ``We do not take even our own observations quite seriously, or accept them as scientific observations, until we have repeated and tested them. Only by such repetitions can we convince ourselves that we are not dealing with a mere isolated `coincidence', but with events which, on account of their regularity and reproducibility, are in principle inter-subjectively testable.''

\subsection{Replication Design}
\begin{recommendation}
    Use the largest feasible replication sample to better distinguish between the null and alternative hypotheses. Treat nominal power targets with caution when power calculations substitute the estimated effect for the true effect: expected replication rates will generally be below the nominal target, even without publication bias or $p$-hacking.
\end{recommendation}

The same-sample-size benchmark used in this paper evaluates replicability under the original research design; it is not intended as a recommendation for replication design. Instead, we recommend that replicators use the largest feasible sample, which maximizes predictive power. This is because larger samples make replication outcomes more informative: significance provides stronger evidence against the null, whereas insignificance provides stronger evidence in its favor.\footnote{Formally, consider the updating factors in equation~\eqref{equation:updating}. Increasing the replication sample size lowers $\left(\frac{r_0(p)}{r(p)}\right)$ following significance and raises $\left(\frac{1-r_0(p)}{1-r(p)}\right)$ following insignificance.}

When evidence informs a concrete decision, the replication sample size should balance the value of more precise evidence against the cost of collecting additional data. In principle, this requires weighing sampling costs against the consequences of adopting an ineffective intervention or rejecting a beneficial one.

Finally, we recommend caution when using the conventional approach of choosing a replication sample size to detect the original effect estimate with, for example, 80\% or 90\% power. Such calculations treat the original estimate as exactly equal to the true effect and ignore its sampling uncertainty. Due to the nonlinearity of the power function, \citet{Vu2024} shows that this approach delivers expected power strictly below the nominal target, even in the absence of $p$-hacking and publication bias.\footnote{The intuition is as follows. Let $T$ denote the nominal power target and $g(\cdot)$ denote the probability of replication. Suppose $\hat{\beta}\sim N(\beta,\sigma^2)$. Because $g(\cdot)$ is locally concave, Jensen's inequality implies that $\mathbb{E}[g(\hat{\beta})]<g(\mathbb{E}[\hat{\beta}])=g(\beta)=T$. Thus, even though the original estimates are unbiased, expected replication power is below the nominal target $T$.} Of course, outcomes from such replications can still be used to rationally update beliefs about the null. The concern is rather that stated replication-rate targets may create the mistaken impression that they are attainable absent research distortions, which is not the case.


\subsection{Encourage and Reward Replications}
\begin{recommendation}
Encourage the publication and citation of replications to establish empirical regularities. Preregistration and pre-analysis plans remain valuable for publishing null results and limiting $p$-hacking, but significant findings can still have low replicability even in the absence of $p$-hacking and publication bias.
\end{recommendation}

Our finding of low replicability at conventional thresholds reinforces recent calls for greater incentives to conduct, publish, and recognize replications \citep{Nosek2012, Coffman2015, Brodeur2023}. As shown in Subsection~\ref{subsection:interpreting_replication_outcomes}, successful replications can generate substantial information about the null hypothesis and hence provide considerable value for subsequent researchers and decision-makers. 

Despite this, replications offer less novelty and professional reward than original studies. To address this imbalance, \citet{Coffman2017} make two proposals. First, leading journals could publish short reports summarizing the replication evidence for influential studies, including replications otherwise buried within broader papers. Second, researchers could cite replications alongside original studies, ensuring that credit and attention reflect the cumulative evidence rather than the first published result alone. Complementing these proposals, the Replication Games of \citet{Brodeur2023} use team-based workshops and joint meta-papers to coordinate replication efforts and reward participants.

Finally, it is important to emphasize that preregistration, pre-analysis plans, and registered reports can limit publication bias and $p$-hacking, but that low replicability can persist even without these distortions. Hence, improving replicability and limiting publication bias are distinct objectives. Three observations underscore this insight. First, the expected replicability for a finding with a given $p$-value is invariant to publication bias and instead determined by the underlying distribution of true effects (Subsection \ref{subsection:replication_power_curve}). Second, Section~\ref{section:nonparametric} finds higher replicability in observational studies than RCTs, despite stronger evidence of $p$-hacking and publication bias in the former \citep{Brodeur2020, Kranz2022}. This occurs due to larger samples in observational studies, which lead to higher replicability. Third, the impact of imposing more stringent publication thresholds would be to mechanically increase replicability, since this favors findings with higher statistical power. In other words, increasing selectivity \textit{improves} replicability. For these reasons, reforms to reduce selective reporting and promote systematic replication should be pursued as complementary objectives.

\section{Conclusion}\label{section:conclusion}
Null hypothesis significance testing has been central to scientific practice -- and the subject of sustained debate about its merits and limitations -- for nearly a century \citep{Fisher1925, NeymanPearson1933, Rozeboom1960, Cohen1994, Wasserstein2019}. Instead of comparing it to proposed alternatives, this paper evaluates the conventional practice of significance testing on its own terms. Using an empirically validated measure of predictive power across fields and research designs, we find that results meeting conventional significance thresholds have relatively low chances of remaining significant in replications of the same sample size. The mechanism is low power in original studies rather than publication bias or $p$-hacking. 

These findings suggest that conventional significance thresholds provide a much weaker guarantee of replicability than is often assumed. A statistically significant result from a single study should instead be treated as suggestive rather than decisive evidence. Robust scientific conclusions require evidence that withstand repeated testing across independent studies.

\bigskip
\bibliographystyle{aer}
{\small
\bibliography{References}
}

\appendix
\renewcommand{\thefigure}{\thesection\arabic{figure}}
\renewcommand{\thetable}{\thesection\arabic{table}}

\section{Proofs}\label{appendix:proofs}
\setcounter{table}{0}
\setcounter{figure}{0}

\textbf{Proof of Theorem~\ref{theorem:predictive_power}}: Let $f_{z \mid |\hat z|,D}(\cdot \mid |\hat z|,1)$ denote the conditional density of $z$ given $|\hat z|$ and publication, $D=1$. Then
\begin{align*}
\Pr(R=1 \mid |\hat z|=x,D=1)
&= \int_{\mathbb R}
\Pr(R=1 \mid |\hat z|=x,z,D=1)\,
f_{z \mid |\hat z|,D}(z \mid x,1)\,dz
\end{align*}

Conditional on $|\hat z|=x$ and $z$, the original test statistic can take either value $\hat z=x$ or $\hat z=-x$. Replication success requires statistical significance in the same direction as the original estimate. Using the Law of Total Probability, 
\begin{align*}
\Pr(R=1 \mid |\hat z|=x,z,D=1)
&=
\Pr(\hat z_r>1.96 \mid z)
\Pr(\hat z=x \mid |\hat z|=x,z,D=1) \\
&\qquad+
\Pr(\hat z_r<-1.96 \mid z)
\Pr(\hat z=-x \mid |\hat z|=x,z,D=1) \\
&=
\frac{
[1-\Phi(1.96-z)]s(x)\varphi(x-z)
+
\Phi(-1.96-z)s(-x)\varphi(x+z)
}{
s(x)\varphi(x-z)+s(-x)\varphi(x+z)
} \\
&=
\frac{
[1-\Phi(1.96-z)]\varphi(x-z)
+
\Phi(-1.96-z)\varphi(x+z)
}{
\varphi(x-z)+\varphi(x+z)
}
\end{align*}

Conditional on $|\hat z|=x$, the original test statistic is either $x$ or $-x$. The second conditional probability in each product gives the probability of each realization conditional on $z$ and publication. The final equality follows from the symmetry of the publication function: $s(x)=s(-x)$.

Next, Bayes' rule implies
\begin{align*}
f_{z \mid |\hat z|,D}(z \mid x,1)
&=
\frac{
f_{D \mid |\hat z|,z}(1 \mid x,z)\,
f_{|\hat z| \mid z}(x \mid z)\,
\pi(z)
}{
\int_{\mathbb R}
f_{D \mid |\hat z|,z'}(1 \mid x,z')\,
f_{|\hat z| \mid z'}(x \mid z')\,
\pi(z')\,dz'
} \\
&=
\frac{
s(x)\left[\varphi(x-z)+\varphi(x+z)\right]\pi(z)
}{
\int_{\mathbb R}
s(x)\left[\varphi(x-z')+\varphi(x+z')\right]\pi(z')\,dz'
} \\
&=
\frac{
\left[\varphi(x-z)+\varphi(x+z)\right]\pi(z)
}{
\int_{\mathbb R}
\left[\varphi(x-z')+\varphi(x+z')\right]\pi(z')\,dz'
}
\end{align*}
\noindent
where $f_{D \mid |\hat z|,z}(1 \mid x,z)=s(x)$ since publication depends on $\hat z$ through $s(\hat z)$ and symmetry implies $s(x)=s(-x)$.

Combining the two above expressions and using $cv(p)\equiv |\hat z|=\Phi^{-1}(1-p/2)$ gives
{\footnotesize
\begin{align*}
\Pr(R=1 \mid p,D=1)
&=
\frac{
\int_{\mathbb R}
[1-\Phi(1.96-z)]\varphi(cv(p)-z)\pi(z)\,dz
+
\int_{\mathbb R}
\Phi(-1.96-z)\varphi(cv(p)+z)\pi(z)\,dz
}{
\int_{\mathbb R}
\left[
\varphi(cv(p)-z)
+
\varphi(cv(p)+z)
\right]
\pi(z)\,dz
}
\end{align*}
}

For the second term in the numerator, let $u=-z$. Then
\begin{align*}
\int_{\mathbb R}
\Phi(-1.96-z)\varphi(cv(p)+z)\pi(z)\,dz
&=
\int_{\mathbb R}
[1-\Phi(1.96-u)]\varphi(cv(p)-u)\pi(-u)\,du
\end{align*}
Similarly,
\begin{align*}
\int_{\mathbb R}
\varphi(cv(p)+z)\pi(z)\,dz
&=
\int_{\mathbb R}
\varphi(cv(p)-u)\pi(-u)\,du
\end{align*}
Recall the symmetrized distribution is defined as $\tilde{\pi}(z)\equiv \frac{\pi(z)+\pi(-z)}{2}$. Relabelling $u$ as $z$, we obtain
\begin{align*}
\Pr(R=1 \mid p,D=1)
&=
\frac{
\int_{\mathbb R}
[1-\Phi(1.96-z)]
\varphi(cv(p)-z)
[\pi(z)+\pi(-z)]\,dz
}{
\int_{\mathbb R}
\varphi(cv(p)-z)
[\pi(z)+\pi(-z)]\,dz
} \\
&=
\frac{
\int_{\mathbb R}
[1-\Phi(1.96-z)]
\varphi(cv(p)-z)
\tilde{\pi}(z)\,dz
}{
\int_{\mathbb R}
\varphi(cv(p)-z)
\tilde{\pi}(z)\,dz
}
\end{align*}
which completes the proof. \qed

\begin{lemma}[Predictive Power Under the Fractional Power Rule]\label{lemma:posterior_power_fpr}
Let $s(\cdot)$ be symmetric about zero and strictly positive. Suppose the replication standard error is chosen according to the fractional power rule with fraction $\psi \in (0,1]$ of the original effect and nominal target power $1-\eta$:
\begin{align*}
\sigma_r = \frac{\psi |\hat z|}{1.96-\Phi^{-1}(\eta)}\sigma
\end{align*}
Then predictive power is
\begin{align}
r(p)
=
\frac{
\int_{\mathbb R}
\left[
1-\Phi\left(
1.96-z\frac{1.96-\Phi^{-1}(\eta)}{\psi cv(p)}
\right)
\right]
\varphi(cv(p)-z)\widetilde{\pi}(z)\,dz
}{
\int_{\mathbb R}
\varphi(cv(p)-z)\widetilde{\pi}(z)\,dz
}
\label{equation:posterior_power_fpr}
\end{align}
where $cv(p)\equiv\Phi^{-1}(1-p/2)$ and $\widetilde{\pi}(z)\equiv \frac{\pi(z)+\pi(-z)}{2}$. The common power rule is the special case $\psi=1$.
\end{lemma}

\begin{proof}
Let $x\equiv|\hat z|$. Under the fractional power rule,
\begin{align*}
z_r \equiv \frac{\beta}{\sigma_r}=
\frac{\beta}{\sigma}\frac{\sigma}{\sigma_r} =
z\frac{1.96-\Phi^{-1}(\eta)}{\psi x}
\end{align*}
Hence, conditional on $z$ and $|\hat z|=x$, the replication test statistic, $\hat z_r$ has mean $z\frac{1.96-\Phi^{-1}(\eta)}{\psi x}$.

Using identical arguments to the proof of Theorem~\ref{theorem:predictive_power}, we obtain Equation~\eqref{equation:posterior_power_fpr}. Setting $\psi=1$ gives the common power rule as a special case.
\end{proof}

\begin{lemma}[Existence and Uniqueness of $p$-value Cutoff for Target Power]\label{lemma:uniqueness} Let $s(\cdot)$ be symmetric about zero and strictly positive, and assume
that $\widetilde{\pi}(\cdot)$ is non-degenerate. Define

$\underline r \equiv \lim_{p\to 1}r(p)$ and
$\overline r \equiv \lim_{p\to 0}r(p)$. Then, for every target
$T\in(\underline r,\overline r)$, there exists a unique
$p^*\in(0,1)$ such that $r(p^*)=T$.

Furthermore, if $\widetilde{\pi}(\cdot)$ has unbounded support above,
then $\overline r=1$. Consequently, every target replication probability
$T\in(\underline r,1)$ is attained by a unique $p$-value cutoff. Finally, $\underline r>1-\Phi(1.96)=0.025$.
\end{lemma}

\begin{proof}
Let $x\equiv cv(p)=\Phi^{-1}(1-p/2)$ and define
\[
g(x)\equiv
\frac{
\overbrace{
\int_{\mathbb R}
\left[1-\Phi(1.96-z)\right]\varphi(x-z)\widetilde{\pi}(z)\,dz
}^{\equiv N(x)}
}{
\underbrace{
\int_{\mathbb R}
\varphi(x-z)\widetilde{\pi}(z)\,dz
}_{\equiv D(x)}
}
\]

By Theorem~\ref{theorem:predictive_power}, $r(p)=g(cv(p))$.

Continuity of $g(x)$ follows from the dominated convergence theorem.
In particular, the integrands defining $N(x)$ and $D(x)$ are continuous
in $x$ and bounded in absolute value by a constant multiple of
$\widetilde{\pi}(z)$. Hence, $N(x)$ and $D(x)$ are continuous, and since
$D(x)>0$, so is $g(x)$.

Next, define
\[
q_x(z)
\equiv
\frac{\varphi(x-z)\widetilde{\pi}(z)}{D(x)}
\]
For any $x_2>x_1$, the likelihood ratio is
\begin{align*}
\frac{q_{x_2}(z)}{q_{x_1}(z)}
&=
\frac{D(x_1)}{D(x_2)}
\frac{\varphi(x_2-z)}{\varphi(x_1-z)} =
\frac{D(x_1)}{D(x_2)}
\exp\left\{
(x_2-x_1)z
-\frac{x_2^2-x_1^2}{2}
\right\}
\end{align*}
This likelihood ratio is strictly increasing in $z$ on the support of
$\widetilde{\pi}(\cdot)$. Hence, since $\widetilde{\pi}(\cdot)$ is
non-degenerate, the distribution $q_{x_2}$ strictly first-order
stochastically dominates $q_{x_1}$. Since $1-\Phi(1.96-z)$ is strictly
increasing in $z$,
\begin{align*}
g(x_2)
&=
\mathbb E_{q_{x_2}}\left[1-\Phi(1.96-Z)\right] >
\mathbb E_{q_{x_1}}\left[1-\Phi(1.96-Z)\right]
=
g(x_1)
\end{align*}
Thus, $g(x)$ is strictly increasing in $x$.

Since $cv(p)=\Phi^{-1}(1-p/2)$ is continuous and strictly decreasing
in $p\in(0,1)$, $r(p)=g(cv(p))$ is continuous and strictly decreasing
in $p$. Its image over $p\in(0,1)$ is therefore
$(\underline r,\overline r)$. For any
$T\in(\underline r,\overline r)$, the intermediate value theorem implies that there exists a $p^*\in(0,1)$ such that $r(p^*)=T$. Strict monotonicity implies $p^*$ is unique.

It remains to establish the endpoint results. Since $q_x(z)\propto \exp\left\{xz-\frac{z^2}{2}\right\}\widetilde{\pi}(z)$, unbounded support of $\widetilde{\pi}(\cdot)$ above implies that $q_x$ shifts arbitrarily far to the right as $x\to\infty$. Hence
$Z\to\infty$ in probability under $q_x$ as $x\to\infty$. Since $1-\Phi(1.96-z)\to1$ as $z\to\infty$ and is bounded, it follows that $g(x)\to1$. Therefore, $\overline r=1$.

Finally, since $cv(p)\to0$ as $p\to1$,
\[
\underline r=g(0)
=
\frac{
\int_{\mathbb R}
[1-\Phi(1.96-z)]\varphi(z)\widetilde{\pi}(z)\,dz
}{
\int_{\mathbb R}
\varphi(z)\widetilde{\pi}(z)\,dz
}
\]
Using the symmetry of $\varphi(z)$ and $\widetilde{\pi}(z)$,
\[
\underline r
=
\frac{
\int_0^\infty
\left\{
[1-\Phi(1.96-z)]
+
[1-\Phi(1.96+z)]
\right\}
\varphi(z)\widetilde{\pi}(z)\,dz
}{
2\int_0^\infty
\varphi(z)\widetilde{\pi}(z)\,dz
}
\]
For every $z\geq0$,
\[
\frac{
[1-\Phi(1.96-z)]
+
[1-\Phi(1.96+z)]
}{2}
\geq
1-\Phi(1.96)
\]
since the left-hand side equals $1-\Phi(1.96)$ at $z=0$ and has derivative
$\frac{\varphi(1.96-z)-\varphi(1.96+z)}{2}>0$ for $z>0$. Since
$\widetilde{\pi}(\cdot)$ is non-degenerate, the inequality is strict. Hence $\underline r>1-\Phi(1.96)=0.025$. This completes the proof.
\end{proof}

\section{Nonparametric Results and Proofs}\label{app:nonparametrics}

In this section we provide general results for nonparametric estimation of replication probabilities. First we formalize and define some key quantities. 

 Let $\hat{z}=z+V$ where $V\sim N(0,1)$ and $z$ has unknown density $\pi$. Consider a replication where sample sizes are multiplied by $\gamma^2$. So $\hat{z}_r=\gamma z+V_r$ where $V_r\sim N(0,1)$.  The random variables $(V,V_r,Z)$ are mutually independent.

  The meta-analyst wishes to estimate the replication probability $r(p,\gamma)$:
 \begin{align*}
        r(p,\gamma) &\equiv \mathbb{P}\left[|\hat z_r|>c \text{ and sign}(\hat z_r)=\text{sign}(\hat z) \: \mid \: |\hat z|=cv(p),D=1\right] \\
        &= \frac{\int_{\mathbb{R}} [1-\Phi(c-\gamma z)] \varphi(cv(p) -z )\tilde{\pi}(z)dz }{\int_{\mathbb{R}} \varphi(cv(p) -z' )\tilde{\pi}(z')dz'}
 \end{align*}
 \noindent where $\varphi$ is the normal density, $\Phi$ is the normal CDF, $\tilde{\pi}(z)\equiv \frac{\pi(z)+\pi(-z)}{2}$ is the symmetrized density, and $cv(p)$ is the two-sided critical value associated with $p$-value $p$.

The meta-analyst observes a random sample of $|\hat{z}|$---i.e. they observe only absolute values. Moreover, they only observe published $|\hat z|$. Let $D=1$ indicate that a finding is published, and suppose its publication probability depends on the observed $t$-ratio through the even selection function $s_{\theta_0}$:
\begin{equation}\label{eq:w_definition}
    \Pr(D=1\mid \hat z )=s_{\theta_0}(\hat z)
\end{equation}
\noindent where the functional form of $s_\theta$ is known but the parameter vector $\theta_0$ is unknown.

\subsection{Identification}

Theorem \ref{thm:identification_nonparm} shows that whenever the publication bias parameter $\theta_0$ is identified then the symmetrized latent density $\tilde \pi$ is identified and therefore the target parameter $r(p,\gamma)$ is also identified. Moreover, the fact that only the absolute values of test statistics are observed by the meta-analyst means that $\pi$ itself is unidentified. For step functions, the identification of $\theta_0$ comes immediately from caliper ratios.
 \begin{theorem}\label{thm:identification_nonparm}
     Assume that the density $\pi$ exists, that publication bias follows Equation (\ref{eq:w_definition}), that the function family $s_\theta(z)$ is known, $\inf_{\theta \in \Theta,z\in \mathbb{R}} s_{\theta}(z)>0$, $s_{\theta}(z)$ is even in $z$, and that $\theta_0$ is identified from the distribution of published $|\hat z|$. Then:
    \begin{itemize}
      \item  $\pi(z)$ is not identified from the distribution of published $|\hat z|$.
        \item   $\tilde{\pi}(z)\equiv \frac{\pi(z)+\pi(-z)}{2}$ is identified  from the distribution of published $|\hat z|$. 
    \end{itemize}
 \end{theorem}
 \begin{proof}
     {\bf Non-Identification of ${\pi}$:} Let $S$ be a Rademacher random variable independent of $z$. Consider two variables $z,z'$. $z$ is positive almost surely and $z'\equiv Sz$. So $z,z'$ have different distributions but yield the same distribution for $|\hat z|$. So  ${\pi}$ is not identified from the distribution of $|\hat z|$.

{\bf Identification of $f_{|\hat z|}$:} Let $ f_{|\hat z|}$ denote the PDF of $|\hat z|$. The researcher observes a random sample of published $|\hat z|$. So the conditional density $f_{|\hat z| \mid D=1}$  is identified. By Bayes' Rule: $f_{|\hat{Z}|}(x) = \frac{f_{|\hat{Z}|\mid D=1}(x)\mathbb{E}[s_{\theta_0}(\hat z)]}{s_{\theta_0}(x)}$. By Bayes' rule: $\mathbb{E}[s_{\theta_0}(\hat{z})] = 1/\mathbb{E}[1/s_{\theta_0}(\hat{z})|D=1]$. So $f_{|\hat{z}|}(x) = \frac{f_{|\hat{z}|\:|D=1}(x)}{s_{\theta_0}(x)\mathbb{E}[1/s_{\theta_0}(\hat{z})|D=1]}$. By assumption the function family $s_\theta$ is known and even in $z$. So if $\theta_0$ is identified then the unconditional density $f_{|\hat{z}|}$ is also identified.

{\bf Identification of $\tilde f_{\hat z}$:} Define $\tilde f_{\hat z}(x) \equiv  \frac{f_{\hat z}(x)+f_{\hat z}(-x)}{2}$. This is the symmetrized version of the PDF of $\hat z$. Since $\tilde f_{\hat z}(x) \equiv  f_{|\hat z|}(|x|)/2$ and we already showed that $f_{|\hat z|}$ is identified, then $\tilde f_{\hat z}$ is also identified. 
      
{\bf Identification of $\tilde{\pi}$:} Let  $S$ be  a Rademacher random variable independent of $(z,V,V_r,\hat z)$. Then $S|\hat z|$ and $S \hat z$ both share the same PDF $ \tilde f_{\hat z}$. Since we already showed that $ \tilde f_{\hat z}$ is identified, the distribution of $S \hat z$ is identified. Notice that $S\hat z = Sz+SV\overset{d}{=} Sz+V$. Since its density is identified, the characteristic function $\phi_{S\hat{z}}(\zeta)$ is identified. By the deconvolution theorem,  $\phi_{Sz}(\zeta)\exp(-\zeta^2/2) = \phi_{S\widehat{z}}(\zeta) $ where $\phi_{Sz}$ is the characteristic function of $Sz$ and $\phi_{S\hat{Z}}$ is the characteristic function of $S\hat{z}$. Since $\exp(-\zeta^2/2)>0$ for all $\zeta$, the distribution of $Sz$ is identified. Since $Sz$ has PDF $\tilde \pi$, so $\tilde \pi$ is identified.
 
 \end{proof}

Since the replication probability $r(p,\gamma)$ is simply a known integral over $\tilde \pi$ we have the corollary:
 \begin{corollary}
   Let the assumptions of Theorem \ref{thm:identification_nonparm} hold. Then, for all $\gamma>0,p\in (0,1)$, the replication probability  $r(p,\gamma)$ is identified  from the distribution of published $|\hat z|$.
 \end{corollary}

\subsection{Estimation}

Define $\mu_j \equiv 2^{-j/2}b_j$. For any $\gamma >0$, the meta-analyst uses an estimator of the same form as before:
\begin{align*}
   \widehat{r}^{(np)}(p,\gamma) &\equiv 1 - \frac{ \sum_{j=0}^{J_n} a_j2^{j/2}  \widehat{\mu}_j  }{  \sum_{j=0}^{J_n} d_j2^{j/2}  \widehat{\mu}_j }\\
    \widehat{\mu}_j
    &\equiv
    \frac{1}{n}\sum_{i=1}^n
    \frac{
         \psi_j\left(cv(p)-|\hat z_i|\right)
        \varphi\left(cv(p)-|\hat z_i|\right) +\psi_j\left(cv(p)+|\hat z_i|\right)
        \varphi\left(cv(p)+|\hat z_i|\right) 
    }{
        s_{\widehat{\theta}_n}\left(\hat z_i\right)
    }
\end{align*}

But, the known deterministic scalar weights $a_j,d_j$ are now the known deterministic constants:
\begin{align*}
   a_j &\equiv  \int_{\mathbb{R}}  \varphi(u)\exp(-u^2/2)\Phi(c-\gamma cv(p)+\sqrt{2}\gamma  u  )\psi_j(u)du \\
    d_j &\equiv \int_{\mathbb{R}}  \varphi(u)\exp(-u^2/2)\psi_j(u)du  
\end{align*}

 Theorem \ref{thm:nonparametric_consistency} finds the rate of consistency of  $\widehat{r}^{(np)}(p,\gamma)$ and shows that this rate slows down as $\gamma$ gets larger.

\begin{theorem}\label{thm:nonparametric_consistency}
   Let the sample $|\hat z_1|,\ldots,|\hat z_n|$ satisfy Assumption \ref{assum:sample}. Assume  that the density $\pi$ exists and has bounded height and that publication bias follows Equation (\ref{eq:w_definition}). Assume that the researcher  has an estimator $\widehat{\theta}_n$ of the publication bias parameters $\theta_0$ such that for some constant $q\in(0,\frac{1}{2}]$:
    $$ \sup_{t\in \mathbb{R}} \left|\frac{1}{s_{\widehat{\theta}_n}(t)} -\frac{1}{s_{{\theta}_0}(t)} \right| = \mathcal{O}_p\left(n^{-q}\right)$$

Then for fixed $p,c,\gamma$ if the researcher sets $J_n = \lceil \frac{2q}{\log 2} \log (A n)\rceil $ for any $A>0$, then:
    $$  \widehat{r}^{(np)}(p,\gamma)  - r(p,\gamma) = \mathcal{O}_p\left(n^{q\log_2\left(\frac{{1+2\gamma^2}}{{2+2\gamma^2}}\right)}(\log n )^{3/4}\right)  $$

If instead $\theta_0$ is known, then:
$$  \widehat{r}^{(np)}(p,\gamma)  - r(p,\gamma) = \mathcal{O}_p\left(n^{\frac{1}{2}\log_2\left(\frac{{1+2\gamma^2}}{{2+2\gamma^2}}\right)}(\log n )^{3/4} \right)  $$
   
\end{theorem}
\begin{proof}
    
Define the inner products:
\begin{align*}
     \langle g_1, g_2 \rangle_{T} &\equiv  \int_{-\infty}^\infty g_1(x)g_2(x)\varphi(x)dx \\
     \langle g_1,g_2 \rangle_{Z} &\equiv  \int_{-\infty}^\infty g_1(x)g_2(x)\varphi(x/\sqrt{2})/\sqrt{2}dx 
\end{align*}

Recall that if $S$ is an independent Rademacher, then $S|\hat z|  \overset{d}{=} S\hat z  \overset{d}{=} Sz +V $ where $V \sim N(0,1)$ and $V \independent Sz$. $Sz$ has PDF $\tilde \pi$ and $S|\hat z|,S\hat z $ share the PDF $\tilde f_{\hat z}$ which is the symmetrized version of the PDF of $\hat z$. 


Since $ \sup_j|\langle  \chi_j ,\tilde \pi \rangle_{Z}|<\infty$ by \cite{faridani2026testingunderpoweredliteratures}, $|\langle \psi_j ,\tilde f_{\hat{Z}} \rangle_T| = \mathcal{O}\left(2^{-j/2}\right)$.

Substituting $z=cv(p)-u$ we have:
\begin{align*}
    r(p,\gamma )  &= \frac{\int_{\mathbb{R}} [1-\Phi(c-\gamma z)] \varphi(cv(p) -z )\tilde\pi(z)dz }{\int_{\mathbb{R}} \varphi(cv(p) -z' )\tilde \pi(z')dz'}\\
    &= 1 - \frac{\int_{\mathbb{R}}  \varphi(cv(p) -z ) \tilde \pi(z) \Phi(c-\gamma z) dz }{\int_{\mathbb{R}} \varphi(cv(p) -z' )\tilde \pi(z')dz'}\\
    &= 1 - \frac{\int_{\mathbb{R}}  \varphi(u ) \tilde \pi(cv(p)-u) \Phi(c-\gamma cv(p)+\gamma u) du }{\int_{\mathbb{R}} \varphi(u' )\tilde \pi(cv(p)-u')du'}\\
    &= 1 - \frac{\int_{\mathbb{R}}  \varphi(u ) \tilde \pi_{p}(u) \Phi(c-\gamma cv(p)+\gamma u) du }{\int_{\mathbb{R}} \varphi(u' )\tilde \pi_{p}(u')du'} \\
    &= 1 - \frac{\int_{\mathbb{R}}  \varphi(u/\sqrt{2} ) \tilde \pi_{p}(u ) \exp(-u^2/4)\Phi(c-\gamma cv(p)+\gamma u  ) du }{\int_{\mathbb{R}} \varphi(u/\sqrt{2} )\exp(-u^2/4)\tilde \pi_{p}(u)du}  \label{eq:change}\\
    &= 1 - \frac{\langle \tilde \pi_{p},\exp(-u^2/4)\Phi(c-\gamma cv(p)+\gamma u  )\rangle_{Z} }{\langle \tilde \pi_{p},\exp(-u^2/4)\rangle_{Z}}  
\end{align*}

where $\tilde \pi_{p}(z)\equiv \tilde \pi(cv(p)-z)$. 

Define $\chi_j(x) \equiv \frac{\He{j}{x/\sqrt{2}}}{\sqrt{j!}}$. Then since the $\chi_j$ form an orthonormal basis:
\begin{align*}
    \exp(-u^2/4)\Phi(c-\gamma cv(p)+\gamma u  ) &= \sum_{j=0}^{\infty} a_j  \chi_j(u)\\
    \exp(-u^2/4) &= \sum_{j=0}^{\infty} d_j  \chi_j(u)\\
    a_j &= \langle  \exp(-u^2/4)\Phi(c-\gamma cv(p)+\gamma u  ),\chi_j\rangle_{Z}\\
     d_j &= \langle  \exp(-u^2/4),\chi_j\rangle_{Z}
\end{align*}

The $\chi_j$ form an orthonormal basis of  the set of all functions with finite norm $\langle  f,f \rangle_Z <\infty $ and the assumed bounded height of $\pi$ guarantees that $\langle \tilde \pi_p,\tilde\pi_p \rangle_Z <\infty $. So by Parseval's Formula there are unknown coefficients $b_j$ such that:
\begin{align*}
    \tilde \pi_p(z) &= \sum_{j=0}^\infty b_j\chi_j(z)\\
   b_j &= \langle \tilde \pi_p, \chi_j\rangle_Z =\mathbb{E}\left[ \varphi((cv(p)-S|z|)/\sqrt{2})\chi_j(cv(p) - S|z|)/\sqrt{2}\right]
\end{align*}

By the orthonormality of $\chi_j$ and the cancelling of common factors:
\begin{align*}
r(p,\gamma)&=  1 - \frac{  \sum_{j=0}^\infty a_j b_j }{\sum_{j=0}^\infty d_j b_j} 
\end{align*}

Again by the orthonormality of the Hermite polynomials:
\begin{align*}
    a_j &= \left\langle  \exp(-u^2/4)\Phi(c-\gamma cv(p)+\gamma u  ), \frac{\He{j}{u/\sqrt{2}}}{\sqrt{j!}} \right\rangle_{Z}\\
    d_j &= \left\langle  \exp(-u^2/4), \frac{\He{j}{u/\sqrt{2}}}{\sqrt{j!}} \right\rangle_{Z}
\end{align*}

Thus the bias of approximating $r(p,\gamma)$ with a sum up to $J_n$ decays exponentially fast in $J_n$:
\begin{align*}
    r(p,\gamma) - \left(1 - \frac{  \sum_{j=0}^{J_n} a_j b_j }{\sum_{j=0}^{J_n} d_j b_j} \right)&= \frac{  \sum_{j=0}^{\infty} a_j b_j }{\sum_{j=0}^{\infty} d_j b_j}- \frac{  \sum_{j=0}^{J_n} a_j b_j }{\sum_{j=0}^{J_n} d_j b_j} \\
  &\equiv \frac{N_0}{D_0}-\frac{N_n}{D_n}\\
  &= \frac{N_0D_n-N_nD_0}{D_0D_n}\\
  &= \frac{N_nD_n+(N_0-N_n)D_n-N_nD_n -N_n(D_0-D_n) }{D_0D_n}\\
  &= \frac{(N_0-N_n)D_n -N_n(D_0-D_n) }{D_0D_n}
\end{align*}

Since $D_0>0$ and $D_n \to D_0$, then for large enough $J_n$, $D_n>D_0/2>0$. Thus, using Lemmas \ref{lem:bound_numerator_coeffs} and \ref{lem:bound_denominator_coeffs} and the fact that the $|b_j|$ are bounded because $\langle\tilde\pi,\tilde\pi\rangle_{Z}<\infty$:
\begin{align*}
\left|\frac{(N_0-N_n)D_n -N_n(D_0-D_n) }{D_0D_n}\right| &\lesssim |N_0-N_n|+|D_0-D_n|\\ 
     \left|  r(p,\gamma) - \left(1 - \frac{  \sum_{j=0}^{J_n} a_j b_j }{\sum_{j=0}^{J_n} d_j b_j} \right)\right| &\lesssim \left|  \sum_{j=J_n+1}^{\infty} a_j b_j\right| +\left|\sum_{j=J_n+1}^{\infty} d_j b_j\right|\\
     &\lesssim  \sum_{j=J_n+1}^{\infty} \left| a_j \right| +\sum_{j=J_n+1}^{\infty} \left|d_j \right|\\
       &\lesssim  \sum_{j=J_n+1}^\infty j^{3/4} \left(\frac{{1+2\gamma^2}}{{2+2\gamma^2}}\right)^{j/2}+\frac{1}{2^{j/2}} 
\end{align*}

The dominant truncation bias comes from the numerator: $ \sum_{j=J_n+1}^\infty j^{3/4} \left(\frac{{1+2\gamma^2}}{{2+2\gamma^2}}\right)^{j/2}$. For ease of notation let $\rho_r \equiv \sqrt{\frac{{1+2\gamma^2}}{{2+2\gamma^2}}}$
\begin{align*}
     \sum_{j=J_n+1}^\infty j^{3/4} \left(\frac{{1+2\gamma^2}}{{2+2\gamma^2}}\right)^{j/2}&=  \sum_{j=J_n+1}^\infty j^{3/4} \rho_r ^j\\
     &=\rho_r^{J_n+1}(J_n+1)^{3/4}\sum_{k=0}^\infty \left(1+\frac{k}{J_n+1}\right)^{3/4}\rho_r ^k\\
     &\leq \rho_r^{J_n+1}(J_n+1)^{3/4}\sum_{k=0}^\infty 2(1+k)\rho_r ^k\\
     &= \rho_r^{J_n+1}(J_n+1)^{3/4}\frac{2}{(1-\rho_r )^2}\\
     &= \mathcal{O}\left( \rho_r ^{J_n}J_n^{3/4}\right)=\mathcal{O}\left( \left(\frac{{1+2\gamma^2}}{{2+2\gamma^2}}\right)^{J_n/2}J_n^{3/4}\right)
\end{align*}

 So the bias decays with:
\begin{align*}
    \left|  r(p,\gamma ) - \left(1 - \frac{  \sum_{j=0}^{J_n} a_j b_j }{\sum_{j=0}^{J_n} d_j b_j} \right)\right| \lesssim \left(\frac{{1+2\gamma^2}}{{2+2\gamma^2}}\right)^{J_n/2}J_n^{3/4}
\end{align*}

The $a_j,d_j$ are known and can be computed via numerical integration but the $b_j$ must be estimated. To do the estimation first notice that:
\begin{align*}
    b_j = \langle \tilde \pi_p,\chi_j \rangle_{Z} = \mathbb{E}\left[ \chi_j(cv(p)- Sz)\varphi((cv(p)- Sz)/\sqrt{2})/\sqrt{2}\right]
\end{align*}

Since the $\sqrt{2}$ appears in both numerator and denominator of $r(p,\gamma)$, it can be ignored from now on. By the Singular Value Decomposition in Example 1 of \cite{CarrascoFlorens} and \cite{faridani2026testingunderpoweredliteratures}:
\begin{align*}
     \mathbb{E}\left[ \chi_j(cv(p)- SZ)\varphi((cv(p)- SZ)/\sqrt{2})\right] &= 2^{j/2}  \mathbb{E}\left[ \psi_j(cv(p)- S\hat{Z})\varphi(cv(p)- S\hat{Z})\right]\\
     &=2^{j/2}\mathbb{E}\left[ \psi_j(cv(p)- S|\hat{Z}|)\varphi(cv(p)- S|\hat{Z}|)\right]\\
     &= \frac{1}{2}2^{j/2}\mathbb{E}\left[ \psi_j(cv(p)- |\hat{Z}|)\varphi(cv(p)- |\hat{Z}|)\right]\\
     &\quad +\frac{1}{2}2^{j/2}\mathbb{E}\left[ \psi_j(cv(p)+ |\hat{Z}|)\varphi(cv(p)+ |\hat{Z}|)\right]
\end{align*}

Substituting and multiplying numerator and denominator by a constant:
\begin{align*}
     r(p,\gamma) &=     1 - \frac{ \sum_{j=0}^\infty a_j2^{j/2}  \mu_j  }{  \sum_{j=0}^\infty d_j2^{j/2}  \mu_j }\\
      \mu_j &\equiv \mathbb{E}\left[ \psi_j(cv(p)- |\hat{Z}|)\varphi(cv(p)- |\hat{Z}|)\right]\frac{1}{\mathbb{P}(D=1)}\\
     &\quad +\mathbb{E}\left[ \psi_j(cv(p)+ |\hat{Z}|)\varphi(cv(p)+ |\hat{Z}|)\right]\frac{1}{\mathbb{P}(D=1)} 
\end{align*}

Therefore, our estimator of $r(p,\gamma )$ is:
\begin{align*}
    \widehat{r}(p,\gamma) &\equiv  1 - \frac{ \sum_{j=0}^{J_n} a_j2^{j/2}  \widehat{\mu}_j  }{\sum_{j=0}^{J_n} d_j2^{j/2}  \widehat{\mu}_j}
\end{align*}

By Lemma \ref{lem:mus}, for any non-negative weights $\alpha_j$ the estimators $\widehat{\mu}_j$ satisfy:
\begin{align*}
   \sum_{j=0}^{J_n} \alpha_j |\widehat{\mu}_j-\mu_j| &\lesssim n^{-q}\sum_{j=0}^{J_n}|\alpha_j|
\end{align*}

So we can bound the stochastic part of the estimation error. 
\begin{align*}
b_j &=  \frac{\mathbb{P}(D=1)}{2\sqrt{2}}2^{j/2}\mu_j\\
    \sum_{j=0}^{J_n} a_j2^{j/2}  \widehat{\mu}_j -\frac{2\sqrt{2}}{\mathbb{P}(D=1)}\sum_{j=0}^{J_n} a_jb_j &=  \sum_{j=0}^{J_n} a_j2^{j/2}  \widehat{\mu}_j -\sum_{j=0}^{J_n} a_j2^{j/2}\mu_j \\
    &=\mathcal{O}_p\left(n^{-q}\sum_{j=0}^{J_n}|a_j|2^{j/2}\right)\\
    &=\mathcal{O}_p\left(n^{-q}\sum_{j=0}^{J_n}j^{3/4}2^{j/2}  \left(\frac{{1+2\gamma^2}}{{2+2\gamma^2}} \right)^{j/2}\right)\\ 
    &= \mathcal{O}_p\left(n^{-q}J_n^{3/4}2^{J_n/2}  \left(\frac{{1+2\gamma^2}}{{2+2\gamma^2}} \right)^{J_n/2}\right)\\
    \sum_{j=0}^{J_n} d_j2^{j/2}  \widehat{\mu}_j - \frac{2\sqrt{2}}{\mathbb{P}(D=1)}\sum_{j=0}^{J_n} d_jb_j &=\mathcal{O}_p\left(n^{-q}\sum_{j=0}^{J_n}2^{j/2}|d_j|\right) =\mathcal{O}_p\left(n^{-q}J_n\right) 
\end{align*}

We already showed that for large enough $J_n$, $\left|\sum_{j=0}^{J_n} d_j2^{j/2}  {\mu}_j\right| > \left|\sum_{j=0}^{\infty} d_j2^{j/2}  {\mu}_j\right|/2$. So with high probability the denominator is bounded away from zero: $$\mathbb{P}\left(\left|\sum_{j=0}^{J_n} d_j2^{j/2}  \widehat{\mu}_j\right| > \left|\sum_{j=0}^{\infty} d_j2^{j/2}  {\mu}_j\right|/4\right)\to 1$$ 
 
The factors of $ \frac{2\sqrt{2}}{\mathbb{P}(D=1)}$ cancel in the fraction. Notice also that the stochastic error of the denominator is dominated. Thus:
\begin{align*}
    \widehat{r}(p,\gamma) - \left(1 - \frac{  \sum_{j=0}^{J_n} a_j b_j }{\sum_{j=0}^{J_n} d_j b_j} \right)  &=  \mathcal{O}_p\left( n^{-q}J_n+n^{-q}J_n^{3/4}2^{J_n/2}  \left(\frac{{1+2\gamma^2}}{{2+2\gamma^2}} \right)^{J_n/2}\right)
\end{align*}

For fixed $\gamma$, the $n^{-q}J_n$ term is dominated. We already bounded the bias. Thus the total estimation error is:
\begin{align*}
     \widehat{r}(p,\gamma)  - r(p,\gamma) &= \mathcal{O}_p\left( \left(\frac{{1+2\gamma^2}}{{2+2\gamma^2}}\right)^{J_n/2}J_n^{3/4}\left( n^{-q} 2^{J_n/2} +1\right)\right)
\end{align*}

For any $A>0$, suppose we set:
\begin{align*}
    J_n &= \lceil \beta \log (An) \rceil
   \end{align*}
Then the $J_n^{3/4}$ factor in the estimation error will only affect the rate up to log factors so we ignore it in the optimization. Now we optimize over $\beta > 0$:

If  $n^{-q} 2^{J_n/2} \lesssim 1$, then $2^{J_n/2}=\mathcal{O}(n^q)$, so $J_n \sim 2\log_2(Dn^q)=\frac{2q}{\log 2} \log (An)$. So the final rate is:
\begin{align*}
     \widehat{r}(p,\gamma)  - r(p,\gamma) &= \mathcal{O}_p\left( \left(\frac{{1+2\gamma^2}}{{2+2\gamma^2}}\right)^{\frac{q}{\log 2} \log (n)}(\log n )^{3/4}\right)\\
     &= \mathcal{O}_p\left(  n^{q\log_2\left(\frac{{1+2\gamma^2}}{{2+2\gamma^2}}\right)}(\log n )^{3/4}\right)
\end{align*}

\end{proof}

\begin{lemma}\label{lem:mus}

Let $q \leq \frac{1}{2}$. Let the assumptions of Theorem \ref{thm:identification_nonparm} hold. For any weights $\alpha_j$ the estimators $\widehat{\mu}_j$ satisfy:
\begin{align*}
   \sum_{j=0}^{J_n} |\alpha_j| |\widehat{\mu}_j-\mu_j| &= \mathcal{O}_p\left( n^{-q}\sum_{j=0}^{J_n}|\alpha_j|\right)
\end{align*}
\end{lemma}
\begin{proof}

\begin{align}
      \mu_j &\equiv \mathbb{E}\left[ \psi_j(cv(p)- |\hat{Z}|)\varphi(cv(p)- |\hat{Z}|)\right]\frac{1}{\mathbb{P}(D=1)}\\
     &\quad +\mathbb{E}\left[ \psi_j(cv(p)+ |\hat{Z}|)\varphi(cv(p)+ |\hat{Z}|)\right]\frac{1}{\mathbb{P}(D=1)}
\end{align}

Define the shorthand $\upsilon(x)$:
\begin{align*}
    \upsilon_j(x) &=  \psi_j(cv(p)- |x|)\varphi(cv(p)- |x|)+ \psi_j(cv(p)+ |x|)\varphi(cv(p)+ |x|)\\
    \mu_j &= \mathbb{E}\left[\upsilon_j(|\hat z|)\right]\frac{1}{\mathbb{P}(D=1)}
\end{align*}

Now add publication bias. Recall that the conditional probability of reporting is equal to $\text{Pr}\left(D=1\:|\: T=t\right) = s_{\theta_0}(t)$. In the presence of publication bias, the $\hat{z}_i$ need to be weighted to remove it using Bayes' Rule:
\begin{align*}
\mathbb{E}\left[ \upsilon_j(|\hat z|)\right]\frac{1}{\mathbb{P}(D=1)}=  \mathbb{E}\left[\frac{\upsilon_j(|\hat z|)}{s_{\theta_0}(\hat{z})}\:|\: D=1\right]
\end{align*}

Since $s$ is even in $\hat z$ by assumption, $s_{\theta}(\hat z) = s_{\theta}(|\hat z|)$. Thus:
\begin{align*}
    \mu_j&= \mathbb{E}\left[\frac{\upsilon_j(|\hat z|)}{s_{\theta_0}(\hat{z})}\:|\: D=1\right]\\
    \widehat{\mu}_j &=   \frac{1}{n}\sum_{i=1}^n \frac{\upsilon_j(|\hat z_i|)}{s_{\widehat{\theta}_n}(|\hat{z}_i|)}
\end{align*}

Suppose we have an estimator $\widehat{\theta}_n$ such that for some fixed $q \in\left(0, \frac{1}{2}\right]$:
\begin{align*}
    \sup_{t\in \mathbb{R}} \left|\frac{1}{s_{\widehat{\theta}_n}(t)} -\frac{1}{s_{{\theta}_0}(t)} \right| = \mathcal{O}_p\left(n^{-q}\right)
\end{align*}

Define $ \tilde{\mu}_j \equiv   \frac{1}{n}\sum_{i=1}^n \frac{\upsilon_j(|\hat z_i|)}{s_{{\theta}_0}(|\hat z_i|)}$. Now we break the estimation error into two pieces. 
\begin{align}\label{eq:twopieces}
    \sum_{j=0}^{J_n}\alpha_j (\widehat{\mu}_j-\mu_j) &= \sum_{j=0}^{J_n}\alpha_j (\tilde{\mu}_j-\mu_j)+\sum_{j=0}^{J_n}\alpha_j (\widehat{\mu}_j-\tilde{\mu}_j)
\end{align}

Now we bound the first sum of (\ref{eq:twopieces}). Notice that since $|\hat z_i|$ are identically distributed, $\mathbb{E}[ \tilde{\mu}_j ] = \mu_j$. By assumption, $\inf_{t,\theta} s_{\theta}(t)>0$ and by \cite{HermiteBound} we have $\sup_{j,x}|\varphi(x)\psi_j(x)|<\infty$. So $\sup_{j,x}|\upsilon_{j}(x)|<\infty$. Moreover, since observations are assumed to be clustered each summand in $\tilde \mu_j$ is correlated with at most $L\geq 1$ others. So there is some constant $C>0$ such that:
\begin{align*}
    \sup_{j\in\{1,2,\ldots\}}\mathbb{E}[|\mu_j-\tilde{\mu}_j |]&\leq    \sup_{j\in\{1,2,\ldots\}}\sqrt{\mathbb{V}[\tilde{\mu}_j]} \leq \frac{C}{\sqrt{n}}\\
    \sum_{j=0}^{J_n} |\alpha_j| \mathbb{E}[|\mu_j-\tilde{\mu}_j |] &\leq \frac{C}{\sqrt{n}}\sum_{j=0}^{J_n} |\alpha_j|
\end{align*}

All moments are implicitly conditional on publication of the sample of $\hat z_i$. By Markov's Inequality:
\begin{align*}
    \sum_{j=0}^{J_n}|\alpha_j| |\tilde{\mu}_j-\mu_j| = \mathcal{O}_p\left( n^{-1/2}\sum_{j=0}^{J_n}|\alpha_j|\right)
\end{align*}

Now we bound the second sum of (\ref{eq:twopieces}). By assumption, $\inf_{t,\theta} s_{\theta}(t)>0$ for all $\theta,t$. Therefore:
\begin{align*}
    \sup_{j\in \{1,\cdots J_n\}}|\widehat{\mu}_j-\tilde{\mu}_j| &=\mathcal{O}_p\left(n^{-q}\right)\\
    \sum_{j=0}^{J_n}|\alpha_j| |\widehat{\mu}_j-\tilde{\mu}_j| &= \mathcal{O}_p\left( n^{-q}\sum_{j=0}^{J_n}|\alpha_j|\right)
\end{align*}

Since $q \leq \frac{1}{2}$, the second sum is the dominant one in (\ref{eq:twopieces}):
\begin{align*}
    \sum_{j=0}^{J_n}|\alpha_j| |\widehat{\mu}_j-\mu_j| = \mathcal{O}_p\left( n^{-q}\sum_{j=0}^{J_n}|\alpha_j|\right)
\end{align*}

\end{proof}

\begin{lemma}\label{lem:bound_numerator_coeffs} Let $\gamma > 0$ be the square root of the ratio of the replication sample size to the original sample size. Then:
   $ |a_j|  \lesssim j^{3/4} \left(\frac{{1+2\gamma^2}}{{2+2\gamma^2}}\right)^{j/2} $ for $j>0$.
\end{lemma}
\begin{proof}
\begin{align*}
        a_j &= \int_{\mathbb{R}}  \varphi(u/\sqrt{2})/\sqrt{2}\exp(-u^2/4)\Phi(c-\gamma cv(p)+\gamma u  )\chi_j(u)du \\
        &= \int_{\mathbb{R}}  \varphi(u)\exp(-u^2/4)\Phi(c-\gamma cv(p)+\sqrt{2}\gamma u  )\chi_j(u \sqrt{2})du \\
        &= \int_{\mathbb{R}}  \varphi(u)\exp(-u^2/2)\Phi(c-\gamma cv(p)+\sqrt{2}\gamma u  )\frac{\He{j}{u}}{\sqrt{j!}}du 
\end{align*}

The generating function of the Hermite polynomials (DLMF 18.12.16) is:
\begin{align*}
    \sum_{j=0}^\infty \He{j}{x} \frac{t^j}{j!} = \exp\left(xt - \frac{t^2}{2}\right)
\end{align*}

Substituting the Hermite generating function in:
\begin{align*}
    F(t) &\equiv \sum_{j=0}^\infty a_j \frac{t^j}{\sqrt{j!}} \\
    &= \int_{\mathbb{R}}  \varphi(u)\exp(-u^2/2)\Phi(c-\gamma cv(p)+\sqrt{2}\gamma u  )  \left(\sum_{j=0}^\infty \frac{\He{j}{u} t^j}{{j!}}\right)du\\
    &=  \int_{\mathbb{R}}  \varphi(u)\exp(-u^2/2)\Phi(c-\gamma cv(p)+\sqrt{2}\gamma u  )   \exp\left(ut - \frac{t^2}{2}\right)du\\
    &= \frac{1}{\sqrt{2\pi}}\int_{\mathbb{R}}  \Phi(c-\gamma cv(p)+\sqrt{2}\gamma u  )   \exp\left(ut - \frac{t^2}{2}-u^2\right)du
\end{align*}

By completing the square: $ut -\frac{t^2}{2}-u^2 = -\left(u-\frac{t}{2}\right)^2-\frac{t^2}{4}$. Thus:
\begin{align*}
    F(t) &= \frac{\exp\left(-t^2/4\right)}{\sqrt{2\pi}}\int_{\mathbb{R}}  \Phi(c-\gamma cv(p)+\sqrt{2}\gamma u  )   \exp\left(-\left(u-\frac{t}{2}\right)^2\right)du
\end{align*}

First consider the case for real $t$. Let the random variable $U$ be distributed: $U\sim N(\frac{t}{2},\frac{1}{2})$. This means that $\sqrt{2}\gamma U \sim N(\frac{\gamma }{\sqrt{2}}t,\gamma^2)$. Thus:
\begin{align*}
    F(t) &= \frac{\exp\left(-t^2/4\right)}{\sqrt{2}}\mathbb{E}\left[  \Phi(c-\gamma cv(p)+\sqrt{2}\gamma U  )   \right]\\
    &= \frac{\exp\left(-t^2/4\right)}{\sqrt{2}}\mathbb{P}\left[Z\leq c-\gamma cv(p)+\sqrt{2}\gamma  U \right]\\ 
    &= \frac{\exp\left(-t^2/4\right)}{\sqrt{2}}\mathbb{P}\left[Z_1\leq c-\gamma cv(p)+\gamma Z_2+ \frac{\gamma t}{\sqrt{2}}\right]\\
   &= \frac{\exp\left(-t^2/4\right)}{\sqrt{2}}\mathbb{P}\left[(Z_1-\gamma Z_2)/\sqrt{1+\gamma^2}\leq (c-\gamma cv(p))/\sqrt{1+\gamma^2}+\frac{t\gamma }{\sqrt{2(1+\gamma^2)}}\right]\\
    &=  \frac{\exp\left(-t^2/4\right)}{\sqrt{2}}\Phi\left( (c-\gamma cv(p))/\sqrt{1+\gamma^2}+\frac{t\gamma }{\sqrt{2(1+\gamma^2)}}\right)
\end{align*}

Since both left and right hand sides are entire functions, the equation also holds for complex $t$. So we have a generating function for the coefficients $a_j$:
\begin{align*}
     F(t) &\equiv \sum_{j=0}^\infty a_j \frac{t^j}{\sqrt{j!}} = \frac{\exp\left(-t^2/4\right)}{\sqrt{2}}\Phi\left(  (c-\gamma cv(p))/\sqrt{1+\gamma^2}+\frac{t\gamma }{\sqrt{2(1+\gamma^2)}}\right)
\end{align*}

This means that the $a_j$ are equal to the derivatives of $F(t)$ evaluated at zero:
\begin{align*}
   a_j &= \frac{F^{(j)}(0)}{\sqrt{j!}}
\end{align*}

So to bound $a_j$ we need only bound the derivatives of $F$. Since the exponential and the normal CDF are entire functions, $F$ is entire. So we can use the Cauchy Differentiation Formula (DLMF 3.4.17) where we integrate around any circle of radius $R$ in the complex plane centered at zero:
\begin{align*}
   F^{(j)}(0) &= \frac{j!}{2\pi i}\int_{|t|=R} \frac{F(t)}{t^{j+1}} dt
\end{align*}

This can be bounded in absolute value by taking the modulus of both sides:
\begin{align*}
    \left|F^{(j)}(0)\right| &\leq j! \frac{1}{R^{j}} \sup_{|t|=R}|F(t)|
\end{align*}

Since $F$ is entire, any $R$ is allowed so we can set $R = \sqrt{j}$. Next we bound $|F(t)|$ in the complex plane. The exponential component can be bounded by $\exp\left(R^2/4\right)$. The normal CDF can be bounded in the following way. For any $a \in \mathbb{R}$ and $\alpha > 0$:
\begin{align*}
    \Phi(x) &= \frac{1}{2}\left(1+\text{erf}(x/\sqrt{2})\right)\\
     \text{erf}\left(x/\sqrt{2}\right)&= \frac{2}{\sqrt{\pi}}\int_{0}^{x/\sqrt{2}}\exp(-t^2)dt\\
    |\text{erf}\left(x/\sqrt{2}\right) | &\leq |x|\exp(Im(x)^2/2)\\
     |\Phi(x)| &\leq 1+|x|\text{exp}(Im(x)^2/2)\\
     \sup_{|t|=R}|\Phi(a+t\alpha)| &\lesssim 1+(|a|+\alpha R)\exp(R^2\alpha^2/2)
\end{align*}

Let $a= (c-\gamma cv(p))/\sqrt{1+\gamma^2}$ and $\alpha = \frac{\gamma }{\sqrt{2(1+\gamma^2)}}$. Substituting into $F$:
\begin{align*}
    \sup_{|t|=R}|F(t)| &\lesssim \sup_{|t|=R}\frac{1}{\sqrt{2}}\left|\exp\left(-t^2/4\right)\Phi\left(  a+\alpha t\right)\right|\\
    &\lesssim   \exp(R^2/4)(1+(|a|+\alpha R)\exp(R^2\alpha^2/2))
\end{align*}

Notice that $\alpha \in (0,1)$. For $R>\max\{|a|/(1-\alpha),2/(1-\alpha)\}$, we have $|a|+\alpha R \leq R$ and:
\begin{align*}
     \sup_{|t|=R}|F(t)| &\lesssim \exp(R^2/4)R\exp(R^2\alpha^2/2)\\
     &= R\exp\left(R^2 (1/4+\alpha^2/2)\right)\\
     &= R\exp\left(R^2/4 \frac{1+\gamma^2+\gamma^2}{(1+\gamma^2)}\right)\\
     &=R\exp\left(R^2 \frac{1+2\gamma^2}{4(1+\gamma^2)}\right)
\end{align*}

Substituting into Cauchy's Differentiation Formula:
\begin{align*}
     \left|F^{(j)}(0)\right| &\leq j! \frac{1}{R^{j}} \sup_{|t|=R}|F(t)|\\
     &\lesssim j! \frac{1}{R^{j}} R\exp\left(R^2 \frac{1+2\gamma^2}{4(1+\gamma^2)}\right)\\
      |a_j| &\lesssim  \sqrt{j!} \frac{1}{R^{j}} R\exp\left(R^2 \frac{1+2\gamma^2}{4(1+\gamma^2)}\right)
\end{align*}

Stirling's approximation says that $\sqrt{j!}\sim \left(\frac{j}{e}\right)^{j/2}(2\pi j)^{1/4} $. So:
\begin{align*}
    |a_j| &\lesssim  \left(\frac{j}{e}\right)^{j/2} j^{1/4}   \frac{1}{R^{j}} R\exp\left(R^2 \frac{1+2\gamma^2}{4(1+\gamma^2)}\right)\\
    &= j^{j/2} j^{1/4}   \frac{1}{R^{j}} R\exp\left(R^2 \frac{1+2\gamma^2}{4(1+\gamma^2)}-\frac{j}{2}\right)
\end{align*}

Setting $R =\sqrt{j}\frac{\sqrt{2+2\gamma^2}}{\sqrt{1+2\gamma^2}}$ (which is allowed because $F$ is entire and eventually exceeds the constant $\max\{|a|/(1-\alpha),2/(1-\alpha)\}$):
\begin{align*}
     |a_j| &\lesssim   j^{j/2} j^{1/4} \sqrt{j}  j^{-j/2} \left(\frac{\sqrt{1+2\gamma^2}}{\sqrt{2+2\gamma^2}}\right)^j\\
     &= j^{3/4} \left(\frac{{1+2\gamma^2}}{{2+2\gamma^2}}\right)^{j/2}
\end{align*}

\end{proof}

\begin{lemma}\label{lem:bound_denominator_coeffs}
   $ |d_j|  \lesssim 2^{-j/2} $ for $j>0$
\end{lemma}
\begin{proof}

By an identical argument to the proof of Lemma \ref{lem:bound_numerator_coeffs}
\begin{align*}
       G(t) &= \sum_{j=0}^\infty d_j \frac{t^j}{\sqrt{j!}} =  \exp\left(-t^2/4\right)/\sqrt{2}\\
       d_j &= \frac{G^{(j)}(0)}{\sqrt{j!}}
\end{align*}

Using the Taylor Expansion:
\begin{align*}
    \exp\left(-t^2/4\right)  &= \sum_{k=0}^\infty\frac{(-1)^k}{4^k k!}t^{2k}
\end{align*}
So
\begin{align*}
    G^{(2k)}(0)&= \frac{(-1)^k(2k)!}{\sqrt{2}4^kk! }\\
    d_{2k}&= \frac{(-1)^k\sqrt{(2k)!}}{\sqrt{2}4^kk! }
\end{align*}

So for odd $j$, $d_j=0$. For even coefficients, notice that the ratio of successive terms is:

\begin{align*}
    \left|\frac{d_{2(k+1)}}{d_{2k}}\right| &= \sqrt{\frac{(2k+2)!}{(2k)!}} \frac{4^kk!}{4^{k+1} (k+1)!}\\
    &= \frac{\sqrt{(2k+2)(2k+1)}}{4(k+1)}\\
    &= \frac{\sqrt{(2k+2)(2k+1)}}{2(2k+2)}\\
    &= \sqrt{\frac{(2k+1)}{4(2k+2)}}\\
    &= \frac{1}{2} \sqrt{\frac{(2k+1)}{(2k+2)}}\\
    &\leq \frac{1}{2}
\end{align*}

So $d_{2j} \lesssim 2^{-j}$ and $d_{j} \lesssim 2^{-j/2}$.

\end{proof}

\section{Additional Empirical Results}\label{appendix:empirical_results}
\setcounter{table}{0}
\setcounter{figure}{0}


\begin{figure} [H]
	\centering
 	\caption{Predictive Power Curves, Systematic Replication Approach}
	\label{figure:posterior_power_curves_repest}
\includegraphics[width=0.65\textwidth]{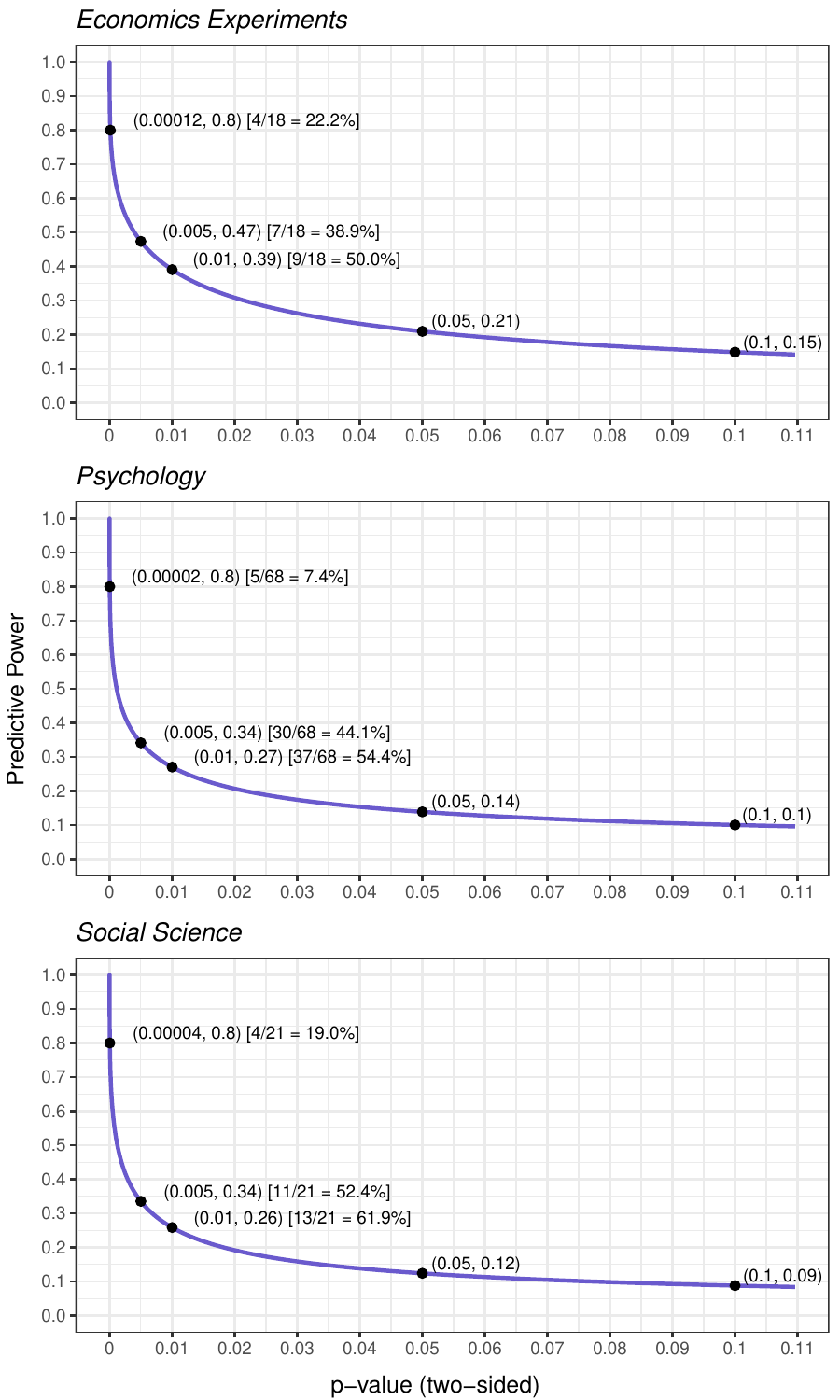} 
 \caption*{\textit{Notes}: The figure plots estimated predictive power curves for experimental economics, psychology, and experimental social science, using the systematic replication approach to estimate $\tilde{\pi}(z)$. Parameter estimates for the distribution of true effects are taken from Table 1 (Economics), Table 2 (Psychology), and Table 11 (Social Science) in \citet{Andrews2019}. All applications assume $\pi(\cdot)$ follows a gamma distribution. Predictive power is the probability that a replication with the same sample size as the original study produces a statistically significant estimate, conditional on the original study's two-sided $p$-value. Square brackets report the percentage of original studies with $p$-values at or below the indicated threshold. Underlying data are from \citet{Camerer2016}, \citet{OpenScience2015}, and \citet{Camerer2018}, respectively.}

\end{figure} 


\begin{figure} [H]
\centering
\caption{Distribution of Predictive Power Conditional on $p=0.05$}
\label{figure:power_bins}
\includegraphics[width=0.65\textwidth]{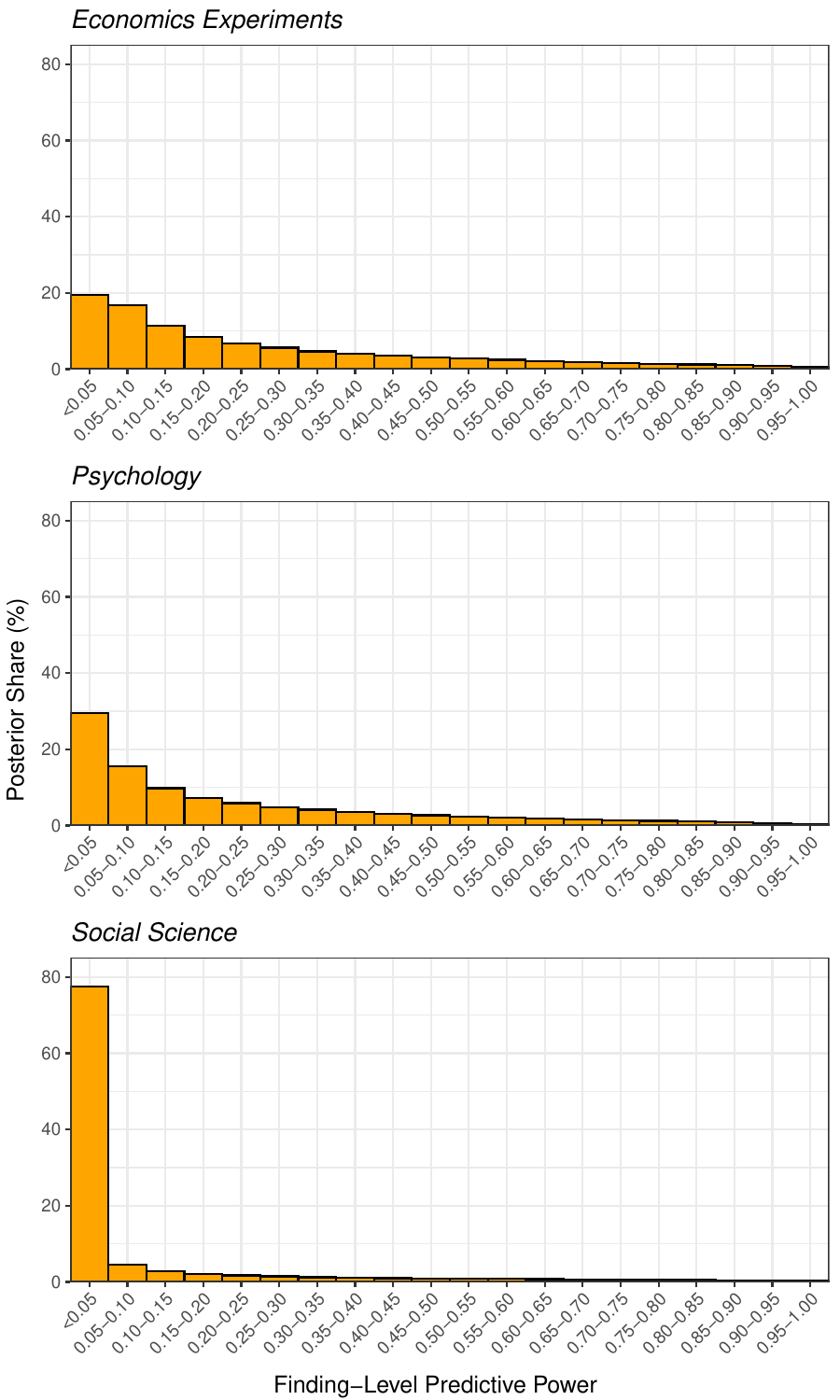}
\caption*{\textit{Notes}: The figure shows the distribution of replication probabilities conditional on an original two-sided $p$-value of 0.05. Each bar reports the share of findings whose probability of a successful replication falls within the indicated decile. A successful replication is statistically significant at the 5\% level and has the same sign as the original estimate. Replications are assumed to use the same sample size as the original study. Estimates of the distribution of true effects are based on the metastudy replication approach. Underlying data are from \citet{Camerer2016}, \citet{OpenScience2015}, and \citet{Camerer2018}, respectively.}
\end{figure}

\begin{figure}[H]
\centering
\caption{Distribution of $z$ Conditional on $p=0.05$}
\label{figure:posterior_z}
\includegraphics[width=0.675\textwidth]{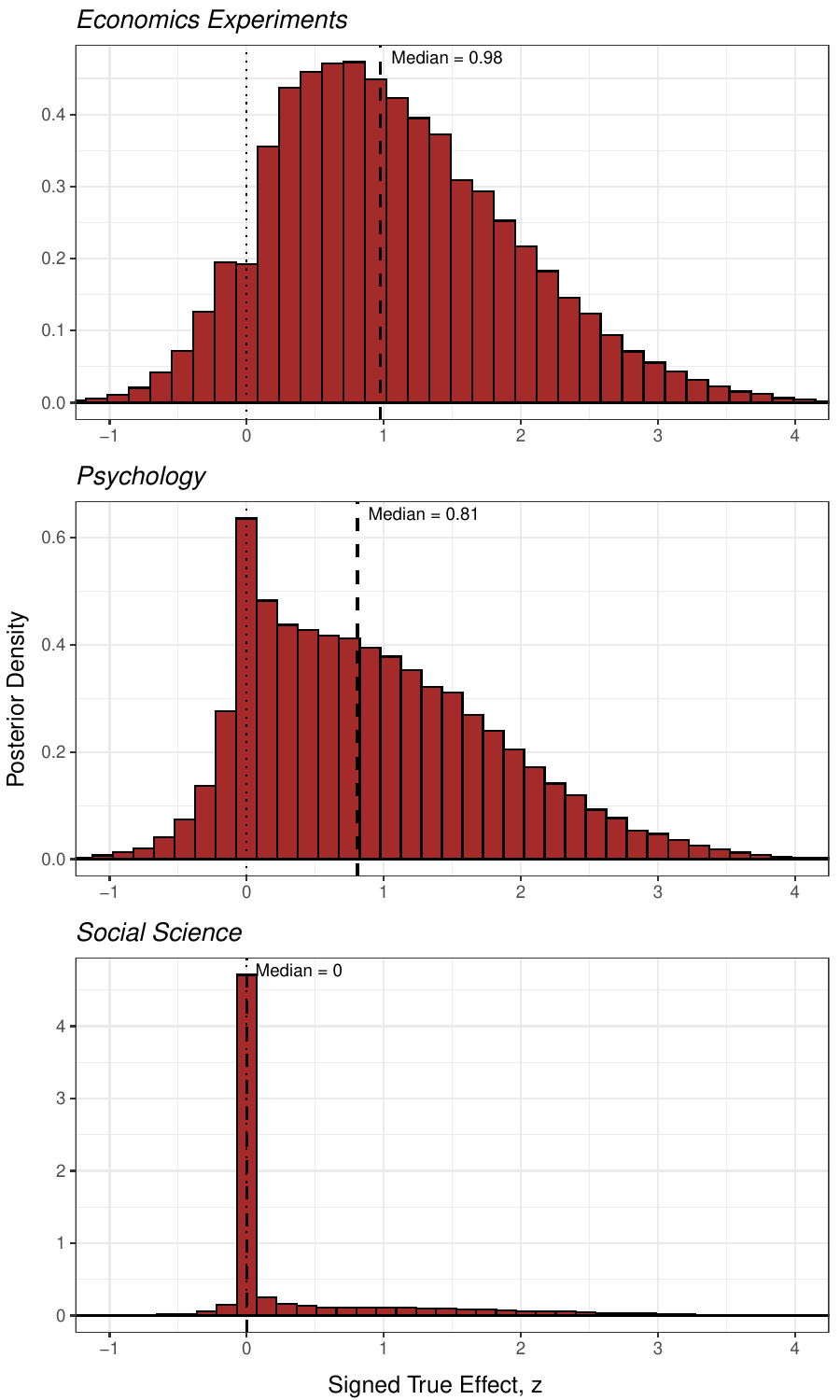}
\caption*{\textit{Notes}: The figure plots the distribution of normalized true effects, $z$, conditional on an original two-sided $p$-value of 0.05. The vertical dotted line marks $z=0$, and the vertical dashed line marks the median of the conditional distribution. Estimates of the distribution of true effects are based on the metastudy replication approach. Underlying data are from \citet{Camerer2016}, \citet{OpenScience2015}, and \citet{Camerer2018}, respectively.}
\end{figure}

\newpage
\section{High-Powered Replications Overstate Replicability}\label{appendix:alternative_rep_sample_size_rules}
\setcounter{table}{0}
\setcounter{figure}{0}

The systematic replication projects used rules to set sample sizes that generally led to larger samples than the original studies. The economics and psychology projects used a common-power rule that treats the original estimate as the true effect and targets a specified nominal power \citep{OpenScience2015, Camerer2016}. Because original estimates are noisy and tend to overstate true effects, expected power is below the nominal target; see \citet{Vu2024}. The social-science project used a fractional-power rule, targeting 90\% power against an effect equal to three-quarters of the original estimate \citep{Camerer2018}.

Table~\ref{table:replication_under_alternative_rules} compares model-based expected replication rates under the implemented high-powered sample-size rules with those under the counterfactual assumption that replications use the same sample size as the original studies. More specifically, for each design, the expected replication rate is calculated as the average across studies, $\frac{1}{N}\sum_{i=1}^N r(p_i;\sigma_{ri})$, where $r(p_i;\sigma_{ri})$ is the expected replication probability for study $i$ given its original $p$-value and the replication standard error implied by the relevant sample-size rule.

Column 3 reports expected replication rates under the implemented higher-powered designs, while Column 4 reports the counterfactual rates if replications used the original sample sizes. The higher-powered designs raise expected replication rates from 0.49 to 0.57 in economics and from 0.46 to 0.53 in psychology and social science. Column 5 shows that, at the study level, these designs increase replication probabilities by 30--36\% on average.\footnote{The final column averages the study-level ratio of high-powered to same-powered replication probabilities, rather than taking the ratio of the averages in Columns 3 and 4.} Thus, systematic replication projects tend to overstate replicability relative to replications conducted at the precision of the original studies.

\begin{table}[H] \centering 
    \scriptsize
  \caption{Expected Replication Rates With Alternative Sample Size Setting Rules} 
  \label{table:replication_under_alternative_rules} 
\begin{tabular}{@{\extracolsep{5pt}} llccc} 
\\[-1.8ex]\hline 
\hline \\[-1.8ex] 
Field & Rule & Implemented Design & Existing Design & Average Ratio \\ 
\hline \\[-1.8ex] 
Economics & Common Power Rule (92\%) & 0.57 & 0.49 & 1.33 \\ 
Psychology & Common Power Rule (92\%) & 0.53 & 0.46 & 1.30 \\ 
Social science & Fractional Power Rule (0.75, 90\%) & 0.53 & 0.46 & 1.36 \\ 
\hline \\[-1.8ex] 
\end{tabular} 
\caption*{\textit{Notes}: The table reports model-based expected replication rates under the implemented replication sample-size rules and under the counterfactual assumption that replications use the same sample size as the original studies. The final column reports the average across studies of the ratio of the implemented-design replication probability to the same-sample-size replication probability.}
\end{table}

\section{Publication Bias Estimates}\label{appendix:publication_probability_ratios}

Table \ref{tab:publication_bias} reports the estimated publication probability ratios corresponding to the estimates in Figure \ref{fig:brodeur-stacked-nonparametric}. Recall the definition of $\theta_m$ from Equation \ref{eq:stepfunction}, which is the probability of publication conditional on $|\hat z|$ lying between the $m$th and $(m-1)$th critical values. Each row reports the estimated ratio of publication probabilities immediately above and below a given critical value. For example, the entry in the first row of the RCT column means that crossing the critical value of $1.64$ increases publication probability by about 8

\begin{table}[ht]
\footnotesize
\centering
\caption{Estimated Publication-Bias Probability Ratios by Research Design}
\label{tab:publication_bias}
\begin{tabular}{lllll}
\hline \hline 
Critical Threshold & RCT & RDD & DID & IV \\ \hline 
1.64 & 1.08 & 0.98 & 1.04 & 1.28 \\ 
  1.96 & 1.03 & 0.97 & 1.02 & 1.27 \\ 
  2.58 & 0.99 & 0.70 & 0.93 & 0.96 \\   \hline 
\end{tabular}
\end{table}

\end{document}